\documentclass[a4paper,10pt]{article}
\usepackage{authblk}
\usepackage{amsmath, latexsym, amssymb, amscd,amsthm}
\numberwithin{equation}{section}
\usepackage{graphicx,float}
\usepackage{pictex}

\newtheorem{theorem}{\bf Theorem}[]

\newtheorem{prop}[]{\bf Proposition}
\newtheorem{corollary}[]{\bf Corollary}

\newtheorem{lem}[]{\bf Lemma}

\begin{document}

\title{Spectral thresholding quantum tomography for low rank states}

\author[1]{Cristina Butucea}
\author[2]{M\u{a}d\u{a}lin Gu\c{t}\u{a}}
\author[2]{Theodore Kypraios}

\affil[1]{Universit\'e Paris-Est Marne-la-Vall\'ee, LAMA(UMR 8050), UPEMLV F-77454, Marne-la-Vall\'ee, France}
\affil[2]{University of Nottingham, School of Mathematical Sciences, University Park, Nottingham NG7 2RD, UK}
\renewcommand\Authands{ and }

\date{}

\maketitle

\begin{abstract}{The estimation of high dimensional quantum states is an important statistical problem arising 
in current quantum technology applications. A key example is the tomography of multiple ions states, employed in the validation of state preparation in ion trap experiments \cite{Haffner2005}. Since full tomography becomes unfeasible even for a small number of ions, there is a need to investigate lower dimensional statistical models which capture prior information about the state, and to devise estimation methods tailored to such models. In this paper we propose several new methods aimed at the efficient estimation of low rank states in multiple ions tomography. All methods consist in first computing the least squares estimator, followed by its truncation to an appropriately chosen smaller rank. The latter is done by setting eigenvalues below a certain ``noise level" to zero, while keeping the rest unchanged, or normalising them appropriately. We show that (up to logarithmic factors in the space dimension) the mean square error of the resulting estimators scales as $r\cdot d/N$ where $r$ is the rank, $d=2^k$ is the dimension of the Hilbert space, and $N$ is the number of quantum samples. Furthermore we establish a lower bound for the asymptotic minimax risk which shows that the above scaling is optimal. The performance of the estimators  is analysed in an extensive simulations study, with emphasis on the dependence on the state rank, and the number of measurement repetitions. We find that all estimators perform significantly better that the least squares, with the ``physical estimator" (which is a bona fide density matrix) slightly outperforming the other estimators.

}
\end{abstract}

\section{Introduction}

Recent years have witnessed significant developments at the overlap between quantum theory and statistics: 
from new state estimation (or tomography) methods  \cite{BlumeKohout,
Ng,Smolin&Gambetta,Heinosaari,
TeoEnglertHradil,
GutaKypraiosDryden,AlquierBut2013}, design of experiments \cite{Smith,Merkel,Nunn}, quantum process and detector tomography \cite{Rahimi,Lundeen} construction of confidence regions (error bars) \cite{BlumeKohout3,Audenaert&Scheel,Christandl&Renner}, quantum tests \cite{Jupp,Temme} entanglement estimation \cite{LandonCardinal}, asymptotic theory \cite{Guta&Kahn2,Hayashi&Matsumoto,Audenaert&Szkola,GutaKiukas}.
The importance of quantum state tomography, and the challenges raised by the estimation of high dimensional systems were highlighted by the landmark experiment \cite{Haffner2005} where entangled states of up to 8 ions were created and fully characterised. 
However, as full quantum state tomography of large systems becomes unfeasible \cite{Monz}, there is significant interest in identifying physically relevant, lower dimensional models, and in devising efficient model selection and estimation methods in such setups  \cite{Gross&Liu&Flammia,Flammia&Gross,Cramer1,GutaKypraiosDryden,AlquierBut2013,Carpentier}. In this paper we reconsider the multiple ions tomography (MIT) problem by proposing and analysing several new methods for estimating low rank states in a statistically efficient way. Below, we briefly review the MIT setup, after which we proceed with presenting the key ideas and results of the paper.

In MIT \cite{Haffner2005}, the goal is to statistically reconstruct the \emph{joint state} of $k$ ions (modelled as two-level systems), from \emph{counts data} generated by performing a large number of measurements on identically prepared systems. The unknown state $\rho$ is  a $d\times d$ density matrix (complex, positive trace-one matrix) where $d= 2^k$ is the dimension of the Hilbert 
space of $k$ ions. The experimenter can measure an arbitrary Pauli observable $\sigma_x, \sigma_y$ or $\sigma_z$ of each ion, \emph{simultaneously} on all $k$ ions. Thus, each measurement setting is labelled by a sequence ${\bf s} = (s_1,\dots ,s_k) \in\{x,y,z\}^k$ out of $3^k$ possible choices. The measurement produces an outcome ${\bf  o} = (o_1,\dots ,o_k) \in \{+1,-1\}^k$, whose probability 
is equal to the corresponding diagonal element of $\rho$ with respect to the orthonormal basis determined by the measurement setting ${\bf s}$. The measurement procedure and statistical model can be summarised as follows. 
For each setting ${\bf s}$ the experimenter performs $n$ repeated measurements and collects the counts of different outcomes 
$N({\bf o} |{\bf s})$, so that the total number of quantum samples used is $N:= n\times 3^k$.  The resulting dataset is a $2^k\times 3^k$ table whose columns are independent and contain all the counts in a given setting. 
A commonly used \cite{Haffner2005} estimation method is \emph{maximum likelihood} which selects the state for which the probability of the observed data is the highest among all states. However, while this method seems to perform well in practice, and has efficient numerical implementations \cite{Hradil}, it does not provide confidence intervals (error bars) for the estimators, and it has been criticised for its tendency to produce rank-deficient states \cite{BlumeKohout}.

The goal of this paper is to find alternative estimators which can be efficiently computed, and work well for \emph{low rank states}. The reason for focusing on low rank states is that they form a realistic model for physical states created in the lab, where experimentalists often aim at preparing  a pure (rank-one) state. While this is generally difficult, the realised states tend to have rapidly decaying eigenvalues, so that they can be well approximated by low rank states. Our strategy is to combine an easy but ``noisy" estimation method -- the least square estimator (LSE) -- with an appropriate 
spectral truncation, tuned using available data only, which sets spurious eigenvalues to zero and allows to reduce the mean square error of the estimator.


The LSE $\widehat{\rho}^{(ls)}_n$ is obtained by inverting the linear map $A:\rho\mapsto \mathbb{P}_\rho$ between the state and the probability distribution of the data, where the unknown probabilities are replaced by the empirical (observed) frequencies of the measurement data. The resulting estimator is unbiased, and is ``optimal" in the sense that it minimises the \emph{prediction error}, i.e. the euclidian distance between the empirical frequencies and the predicted probabilities. 
However, one of the disadvantages of the LSE is that it does not take into account the physical properties of the state, i.e. its positivity and trace-one property. More importantly, as we explain below, the LSE has a relatively large \emph{estimation error} for the class of low rank states, and performs well only on very mixed states. This is illustrated in Figure \ref{fig.ls} where the eigenvalues of 
$\widehat{\rho}^{(ls)}_n$ are plotted (in decreasing order) against those of the true state $\rho$, the latter being chosen to have rank $r=2$. 
We see that while the non-zero eigenvalues of $\rho$ are estimated reasonably well, the LSE is poor in estimating the 
zero-eigenvalues, and as consequence, it has a large estimation variance. 

\begin{figure}[h]
\begin{center}
\includegraphics[width= 6.5cm]{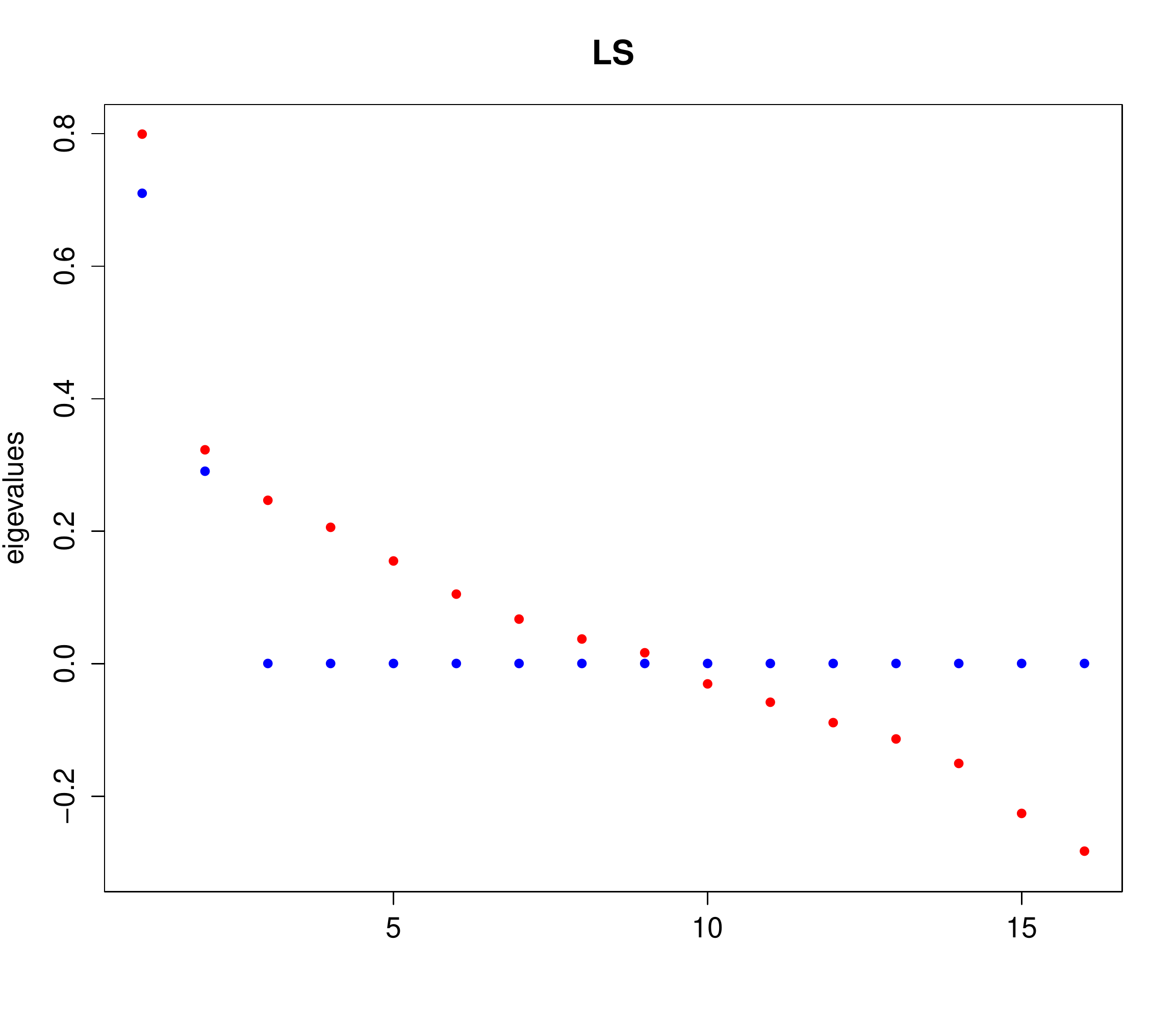}
\includegraphics[width= 6.5cm]{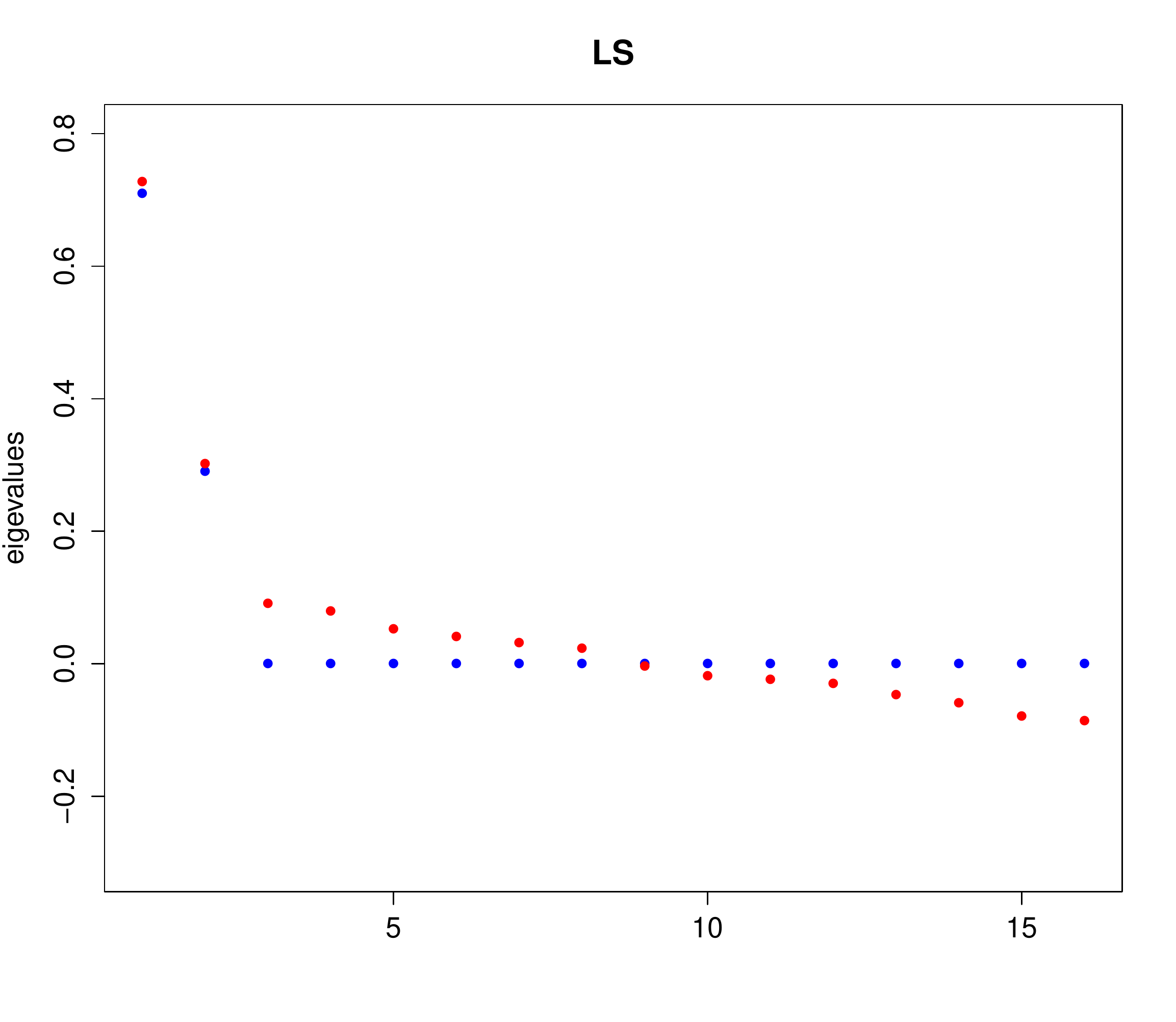}
\caption{Eigenvalues of the LSE (red) arranged in decreasing order, versus those of the true state of $k=4$ ions of rank $r=2$ (blue), for $n=20$ measurement repetitions (LEFT) and $n=100$ measurement repetitions (RIGHT).}\label{fig.ls}
\end{center}
\end{figure}

Our goal is to design more precise estimators, which have the LSE as a starting point, but take into account the ``sparsity" properties of the unknown state. Figure \ref{fig.ls} suggests that  the non-zero eigenvalues of the LSE which are below a certain ``statistical threshold", can be considered as statistical noise and may be set to zero in order to improve the estimation error. To find this noise level, we establish a concentration inequality (see Proposition \ref{prop:linear}) which shows that the operator-norm error $\|\widehat{\rho}^{(ls)}_n -\rho\|^2$ is upper bounded by a rate $\nu^2$ which (up to logarithmic factors in $d$) is proportional to $d/N$.

The first estimator we propose, is a \emph{rank penalised} one obtained by diagonalising the LSE, arranging its eigenvalues in decreasing order of their absolute values, and setting to zero all those eigenvalues whose absolute values  are below the threshold $\nu$
$$
\widehat\rho^{(ls)}_n= \sum_{i=1}^d \widehat \lambda_i | \hat{\psi}_i\rangle \langle \hat{\psi}_i | \quad\longrightarrow\quad
\widehat\rho^{(pen)}_n= \sum_{|\widehat \lambda_i|\geq \nu} \widehat \lambda_i | \hat{\psi}_i\rangle \langle \hat{\psi}_i |.
$$
The same outcome can be obtained as solution of the following penalised estimation problem: 
among all selfadjoint matrices, choose the one that is close to the LSE but in the same time it has low rank, 
so that it minimises over $\tau$ the norm-two square discrepancy penalised by the rank
$$
D(\tau)  := \|\tau - \widehat \rho_n ^{(ls)} \|_2^2 + \nu^2 \cdot {\rm rank}(\tau).  
$$
In particular the estimator's rank is determined by the data. 
In Theorem \ref{thm:nonlinear} we show that if $\rho$ is of \emph{unknown} rank $r\leq d$, then the mean square error (MSE) 
$\mathbb{E}\| \widehat\rho^{(pen)}_n - \rho\|_2^2$ is upper-bounded (up to logarithmic factors) by the rate $ (r \cdot d)/N$. 
This captures the expected optimal dependence on the number of parameters for a state of rank $r$. Indeed, in section \ref{sec.lower.bound} we show that no estimator can improve the above rate for all states of rank $r$, cf. Theorem \ref{th.minmax} for the \emph{asymptotic minimax} lower bound.

The penalised estimator has however the drawback that it may not represent a physical state. To remedy this, and further improve its statistical accuracy, we propose a \emph{physical estimator} which is the solution of the following optimisation problem. We seek the \emph{density matrix} which is closest to the LSE $\widehat{\rho}^{(ls)}_n$, and whose non-zero eigenvalues are larger that the threshold $4\nu$. It turns out that the solution can be found via a simple iterative algorithm whereby at each step the eigevalues of $\widehat{\rho}^{(ls)}_n$ below the threshold are set to zero, 
and the remaining eigenvalues are normalised by shifting with a common constant, while the eigenvectors are not changed throughout the process. In Theorem \ref{thm:recstate} we show that the physical estimator satisfies a similar upper bound to the penalised one. 
   
In section \ref{sec.simulations} we present results of extensive numerical investigations of the two 
proposed estimators. In addition we consider the \emph{oracle} ``estimator", which is simply the spectral truncation of the LSE that is closest to the true state $\rho$, and the \emph{cross-validated} estimator $\widehat{\rho}^{(cv)}_n$ which aims at finding the optimal truncation rank by estimating the Frobenius error by cross-validation. In fact, we found that cross-validation can help in better tuning the constant factor of the threshold rate of the penalised and physical estimators. As expected from the theoretical results, we find that all estimators perform significantly better than the LSE on low rank states; moreover the physical estimator has slightly smaller estimation error than the others, including the oracle estimator. We also find that all methods converge to the correct rank in the limit of large number of repetitions
but through different routes: the penalised estimator tends to underestimate, while the physical one tends to overestimate the rank, for small number of samples.   

Having discussed the upper bounds on the estimators' MSE, we would like to know how they compare with the best possible estimation procedure. One way to characterise the latter is through the \emph{asymptotic minimax risk} for the class of states of a given rank $r$. From asymptotic statistics theory \cite{vanderVaart} we know that for every sequence of estimators 
$\widehat{\rho}_n$, the following lower bound for its asymptotic maximum risk over the set $\mathcal{S}_{d,r}$ of state of rank $r$ holds
$$
\lim\inf_{n\to\infty} 
n \cdot 3^k \sup_{\rho\in \mathcal{S}_{d,r}} \mathbb{E} \|\widehat{\rho}_n- \rho \|_2^2 \geq 
\sup_{\rho\in \mathcal{S}_{d,r}} {\rm Tr}(I(\rho)^{-1} G(\rho)). 
$$
On the right side, $I(\rho)$ is the Fisher information corresponding to all measurement settings taken together, and $G(\rho)$ is a positive matrix describing the quadratic approximation of the Frobenius distance around $\rho$. In Theorem \ref{th.minmax} we show that the right side is lower bounded by $2r(d-r) $ which shows that (up to logarithmic factors) the upper bounds of the penalised and physical estimators have the same scaling as the asymptotic minimax risk.

Recently, a number of papers discussed related aspects of quantum tomography problems. 
The idea of the penalised estimator has been proposed in \cite{AlquierBut2013}, which provided a weaker upper bound for its MSE. Reference \cite{GutaKypraiosDryden} analyses model selection methods for finite rank models and maximum likelihood estimation. Reference \cite{Koltchinsky} proposes a different estimator and establishes a comparable upper bound for its MSE. The class of low rank states is also employed in compressed sensing quantum tomography \cite{Gross&Liu&Flammia,Gross2011,Carpentier}, but their statistical model is based on expectations of Pauli observables rather than measurement counts.

The paper is organised as follows. In section \ref{sec.general} we describe the measurement procedure and introduce 
the statistical model of MIT. In section \ref{sec.ls} we define the  linear (least squares) estimator and derive an upper bound on its operator norm error 
which improves on a previous bound of \cite{AlquierBut2013}. 
In section \ref{section:modified} we define the penalised and threshold estimators and derive upper bounds for their mean square errors with respect to the norm-two square (Frobenius) distance.
The performance of the different methods is analysed in Section \ref{sec.simulations}. 
An asymptotic lower bound for the minimax risk is derived in section \ref{sec.lower.bound}, based on the Fisher information of the measurement data. The upper and lower bounds match in the scaling with the number of parameters and number of total measurements, up to a logarithmic factor. We give a detailed description of the numerical implementation of the algorithms, including the cross-validation routines used for tuning the pre-factor of the penalty and threshold constants. We illustrate the simulation results with box plots of the Frobenius errors for the least squares, oracle, cross-validation, penalisation and threshold estimator, for states of ranks 1,2,6 and 10, and for different choices of measurement repetitions $n=20, 100$. Additionally, we plot the empirical distribution of the chosen rank for different estimators, showing the concentration on the true rank as the number of repetitions increases. 

%

\section{Multiple ions tomography}\label{sec.general}

This paper deals with the problem of estimating the joint quantum state of $k$ two-dimensional systems (qubits), as encountered in ion trap quantum tomography \cite{Haffner2005}. The two-dimensional system is determined by two energy levels of an ion, while the remaining levels can be ignored as they remain unpopulated during the experiment. The joint Hilbert space of the ions is therefore the tensor product 
$\left(\mathbb{C}^2\right)^{\otimes k}\cong \mathbb{C}^d$ where  $d= 2^k$, and the state is a density matrix $\rho$ on this space, i.e. a positive $d\times d$ matrix of trace one.

Our statistical model is derived from standard ion trap measurement procedures, and takes into account the specific statistical uncertainty due to finite number of measurement repetitions. We consider that for each individual qubit, the experimenter can measure one of the three Pauli observables $\sigma_{x},\sigma_{y},\sigma_{z}$. A measurement set-up is then defined by a \emph{setting} ${\bf  s}= (s_1,\dots, s_k)\in \mathcal{S}_{k} :=\{x,y,z\}^{k}$ which specifies which of the 3 Pauli observables is measured for each ion. 
For each fixed setting, the measurement produces random outcomes 
${\bf o} \in \mathcal{O}_{k}:= \{+1,-1\}^{k}$ with probability 
\begin{equation}\label{eq.proba}
p({\bf o}|{\bf s}) := {\rm Tr}(\rho P^{\bf s}_{\bf o} )=
\langle e_{\bf o}^{\bf s}|\rho |e_{\bf o}^{\bf s}\rangle ,
\end{equation}
where $P^{\bf s}_{\bf o}$ are the one-dimensional projections 
\begin{equation}\label{eq.basis.d}
P^{\bf s}_{\bf o} = | e_{o_{1}}^{s_{1}} \rangle\langle e_{o_{1}}^{s_{1}} |
\otimes \dots \otimes
 | e_{o_{k}}^{s_{k}} \rangle\langle e_{o_{k}}^{s_{k}} |
\end{equation}
and $ | e_{o}^{s} \rangle$ are the eigenvectors of the Pauli matrices, i.e.
$
\sigma_{s}  | e_{\pm}^{s}\rangle = \pm\,  | e_{\pm}^{s}\rangle.
$

The measurement procedure consists of choosing a setting ${\bf s}$, and performing $n$ repeated measurements in that setting, on identically prepared systems in state $\rho$. This provides information about the diagonal elements of $\rho$ with respect to the chosen measurement basis, i.e. the probabilities $p({\bf o}|{\bf s})$. In order to identify the other elements, the procedure is then repeated for all $3^k$ possible settings.

Before describing the statistical model of the measurement counts data, we start by discussing in more detail the relation between the unknown parameter $\rho$ and the probabilities $p({\bf o}|{\bf s})$. Consider the ``extended" set of Pauli operators $\{\sigma_x,\sigma_y, \sigma_z, \sigma_I:= \mathbf{1}\}$ which form a basis in $M(\mathbb{C}^2)$. 
We construct the tensor product basis in $M(\mathbb{C}^d)$ with elements 
$ \sigma_{\bf b} = \sigma_{b_1}\otimes \dots\otimes \sigma_{b_k} $ where $ {\bf b} \in \{x,y,z, I \}^k $ and note that the following orthogonality relations hold
$
{\rm Tr} (\sigma_{\bf b} \sigma_{\bf c}) = d\delta_{{\bf b}, {\bf c}}.
$
The state $\rho$ can be expanded in this basis as 
\begin{equation}\label{eq.dec.rho}
\rho= \sum_{{\bf b}\in \{I,x,y,z\}^k} \rho_{\bf b} \sigma_{\bf b}, \qquad \mathrm{where}\quad\rho_{\bf b} = {\rm Tr}(\rho \sigma_{\bf b})/d.
\end{equation}
Equation \eqref{eq.proba} can then be written as
$$
p({\bf o}|{\bf s}) =\sum_{{\bf b} \in \{I,x,y,z\}^k} \rho_{\bf b}{\rm Tr}(\sigma_{\bf b} P_{\bf o}^{\bf s}) = \sum_{{\bf b} \in \{I,x,y,z\}^k} \rho_{\bf b} A_{\bf b} ({\bf o}|{\bf s}).
$$
The coefficients $A_{\bf b}({\bf o}|{\bf s})$ can be  computed explicitly as
\begin{equation}\label{calculA}
A_{\bf b}({\bf o}|{\bf s}) ={\rm Tr}(\sigma_{\bf b} P_{\bf o}^{\bf s})  =\prod_{j \not \in E_{\bf b}} o_j \cdot I(b_j = s_j),
\end{equation}
where $E_{\bf b}:= \{ i:  b_i = I\}$.  Let $\tilde{\rho}\in \mathbb{C}^{4^k}$ be the representation of $\rho$ as a the vector of coefficients $\rho_{\bf b}$, and let 
${\bf p}$ be the corresponding vector of probabilities for all settings $\left( p({\bf o}|{\bf s}) \,:\, ({\bf o},{\bf s}) \in \mathcal{O}_k\times \mathcal{S}_k\right)$, with settings, and outcomes within settings ordered in lexicographical order. The measurement is then described by the  linear map 
${\bf A}:\mathbb{C}^{4^k} \to \mathbb{C}^{3^k}\otimes \mathbb{C}^{2^k}$ with matrix elements  
$A_{\bf b}({\bf o}|{\bf s})$ defined in (\ref{calculA}), such that
\begin{equation}\label{eq.linear.map.a}
{\bf p} = {\bf A}\tilde{\rho}.
\end{equation}
The linear map ${\bf A}$ is injective and its inverse can be computed as 
$
\tilde \rho = ({\bf A}^* \cdot {\bf A})^{-1} \cdot {\bf A}^* \cdot {\bf p} ,
$
which together with the decomposition \eqref{eq.dec.rho} and Lemma \ref{lemmaA} implies
\begin{equation}\label{eq.inverse.linear}
\rho = \sum_{\bf b} \sum_{\bf o} \sum_{\bf s} p({\bf o}|{\bf s}) \frac{A_{\bf b}({\bf o}|{\bf s}) }{2^k 3^{d({\bf b})}} \sigma_{\bf b} =\sum_{\bf b} \tilde\rho_{\bf b}\sigma_{\bf b} , \qquad \mathrm{where} \quad d({\bf b}) := | E_{\bf b}|.
\end{equation}

The above formula allows to reconstruct the matrix elements from the measurement probabilities. However, since the experiment only provides random counts from these probabilities, we need to construct a statistical model for the measurement data. After $n$ repetitions of the measurement with setting ${\bf s}$, we collect independent, identically distributed observations $X_{i|{\bf s}}\in \mathcal{O}_k$, for $i$ from 1 to $n$. The data can be summarised by the set of counts 
$\{ N({\bf o}|{\bf s}) : {\bf o} \in \mathcal{O}_{k}\}$, where $N({\bf o}|{\bf s}) = \sum_i I(X_{i|{\bf s}}={\bf o})$ is the number of times that outcome ${\bf o}$ has occurred. After repeating this for each setting ${\bf s}\in \mathcal{S}_{k}$, we collect all the data in the counts dataset $ D:= \{ N({\bf o}|{\bf s}) :  ({\bf o}, {\bf s}) \in \mathcal{O}_{k} \times \mathcal{S}_{k} \}$. Since successive preparation-measurement cycles are independent of each other, the probability of a certain dataset $D$ is given by the product of multinomials
 \begin{equation}\label{eq.multinomial.prod}
\mathbb{P}_{\rho}(D) = \mathbb{P}_{\rho}\left( \{ N( {\bf o} |{\bf s} ) :  ({\bf o}, {\bf s}) \in \mathcal{O}_{k} \times \mathcal{D}_{k}  \}\right) = 
\prod_{\bf s} 
\frac{n!}{\prod_{{\bf o}} N( {\bf o} | {\bf s} )!} 
\prod_{\bf o} \mathbb{P}_{\rho} ({\bf o}|{\bf s})^{N({\bf o}|{\bf s}) }
\end{equation}
The statistical problem is to estimate the state $\rho$ from the measurement data summarised by the counts dataset $D$. 
The most commonly used estimation method is maximum likelihood (ML). The ML estimator is defined by
$$
\hat{\rho}^{(ml)}_n (D):=\underset{\sigma}{\arg\max}  \,
\mathbb{P}_{\sigma}( D)
$$
where the maximum is taken over all \emph{density matrices} $\tau$ on $\mathbb{C}^d$, and can be computed by using standard maximisation routines, or the iterative algorithms proposed in \cite{Hradil,Rehacek}.  However, ML becomes impractical for about $k=10$ ions, and the iterative algorithm has the drawback that it cannot be adapted to models where prior information about the state is encoded in a  lower dimensional parametrisation of the relevant density matrices, e.g. when the states are low rank. In the next section we discuss an alternative method, the least square estimator, and derive an upper bound on its mean square error. After this, we will show that  by ``post-processing" the least squares estimator using penalisation and thresholding methods, its performance can be considerably improved when the unknown state has low rank.

 \section{The linear (least squares) estimator}\label{sec.ls}

Since the vectorise version $\tilde\rho$ of the state $\rho$ satisfies \eqref{eq.linear.map.a}, it is the solution of the optimisation
$$
\tilde \rho = \arg \inf_{\tilde \tau \in \mathbb{C}^{4^k}} \| {\bf p} - {\bf A} \tilde \tau\|^2,
$$ 
giving $\rho = \sum_{\bf b} \tilde \rho_{\bf b} \cdot \sigma_{\bf b} $ in \eqref{eq.inverse.linear}. If the number of repetitions $n$ is large compared with the dimension $d$, then the outcomes' empirical frequencies are good approximations of the corresponding probabilities, i.e. $f({\bf o}|{\bf s}) :=N({\bf o}|{\bf s})/n \to p({\bf o}|{\bf s})$ by the law of large numbers. Therefore, by replacing ${\bf p}$ by the vector of frequencies ${\bf f}$ in in the previous display, we can define the least square estimator of $\rho$
$$
\tilde \rho_n^{(ls)} := \arg \inf_{\tilde \tau \in \mathbb{C}^{4^k}} \|{\bf f} - {\bf A} \tilde \tau\|^2,
$$
which has the explicit expression
\begin{equation}\label{eq.least.squares}
\widehat \rho^{(ls)}_n 
 =  \sum_{\bf b} \sum_{\bf o} \sum_{\bf s} f({\bf o}|{\bf s}) \frac{A_{\bf b}({\bf o}|{\bf s}) }{2^k 3^{d({\bf b})}} \sigma_{\bf b}.
\end{equation}
Note that in this case it comes down to replacing the unknown probability ${\bf p}$ in equation \eqref{eq.inverse.linear} with the empirical frequencies ${\bf f}$ (also known as the plug-in method).


In spite of this ``optimality" property and its computationally efficiency, the least square estimator has the disadvantage that in general it is not a state, i.e. it is not trace-one and may have negative eigenvalues. A more serious disadvantage is that its risk -- measured for instance by the mean square error $\mathbb{E} ( \| \widehat \rho^{(ls)}_n - \rho \|_2^2)$ -- is large compared with other estimators such as the maximum likelihood estimator. This is due to the fact that the linear estimator does not use the physical properties of the unknown parameter $\rho$, that is positivity and trace-one. As we will see below, the modified estimators proposed in Section~\ref{section:modified} outperform  
the least square while adding only a small amount of computational complexity. Moreover, the second estimator will be a density matrix.

In the remainder of this section we provide concentration bounds on the square error of the linear estimator, which will later be used in obtaining the upper bounds of the improved estimators. The following Proposition improves the rate $k(4/3)^k /n$ obtained in \cite{AlquierBut2013} to $k(2/3)^k/n$. 
\begin{prop}\label{prop:linear}
Let  $\widehat \rho^{(ls)}_n$ be the linear estimator  of $\rho$. Then, for any $\varepsilon >0$ small enough the 
following operator norm inequality holds with probability larger than $1-\varepsilon$ under $\mathbb{P}_\rho$
$$
\|\widehat \rho^{(ls)}_n - \rho\| \leq  \nu(\varepsilon),
$$
where 
$$
\nu(\varepsilon)^2 =  \frac 2{n} \left(\frac 23\right)^{k}\log \left(\frac{2^{k+1}}{\varepsilon}\right) = 
 \frac{2 d}{N}\log\left(\frac{2 d}{\varepsilon}\right)
$$
with $N:= n \cdot 3^k$ the total number of measurements. The same bound holds when $k=k(n)$ as long as $\nu(\varepsilon) \to 0$. 
\end{prop}
\emph{Proof.} See Appendix \ref{app.proof.prop.linear}.

As a side remark we note that projecting the least-squares estimator onto the space of Hermitian matrices with trace 1, does not change the rate of convergence from Proposition~\ref{prop:linear}. The following Proposition allows us to assume that, without loss of generality, the least-squares estimator has also trace 1.
\begin{prop}\label{prop:trace}
Under the notation and assumptions in Proposition~\ref{prop:linear}, let
\begin{equation}\label{rho.tilde.ls}
\widehat \rho_n^{(ls,n)}  = \underset{\tau: tr(\tau)=1}{\arg\min}  \|\tau - \widehat \rho_n^{(ls)} \|_2^2.
\end{equation}
Then with probability larger than $1-\varepsilon$ we have
$
\|\widehat \rho_n^{(ls,n)} - \rho\| \leq 2 \nu( \varepsilon).
$
\end{prop}

\emph{Proof.} See Appendix \ref{app.prop:trace}.

\section{Rank-penalised and threshold projection estimator} \label{section:modified}

In this section we investigate two ways to improve the least-squares estimators. The first method is to project the least-squares estimator onto the space of finite rank Hermitian matrices of an appropriate rank. We prove upper bounds for its risk with respect to the Frobenius (norm-two) distance. Building on the knowledge about the rank-penalised estimator, we define the second estimator which is the projection of the least-squares estimator on the space of physical states whose eigenvalues are larger than a certain positive noise threshold. We give an simple and fast algorithm producing a proper density matrix from the data, which also inherits the good theoretical properties of the rank-penalized estimator.

\subsection{Rank-penalised estimator}
We introduce here the rank-penalised nonlinear estimator, which can be computed from the least-squares estimator by truncation to an appropriately chosen rank.

As noted earlier, while the least-square estimator is unbiased, it has a large variance due to the fact that it does not take into account the physical constraints encoded in the unknown parameter $\rho$.  A possible remedy is to ``project" the least squares estimator onto the space of physical states, i.e. positive, trace-one matrices. This method will be discussed in the following subsection. Another improvement can be obtained by taking into account the ``sparsity" properties of the unknown state. For instance, in many experimental situations the goal is to create a particular low rank, or even pure state. The fact that such states can be characterised with a smaller number of parameters than a general density matrix, has two important consequences. Firstly, they can be estimated by 
measuring an ``informationally incomplete" set of observables, as demonstrated in \cite{Gross&Liu&Flammia,Flammia&Gross}. Secondly, the prior information can be used to design estimators with reduced estimation error compared with generic methods which do not take into account the structure of the state. Roughly speaking, this is because each unknown parameter brings its own contribution to the overall error of the estimator.

However, the downside of working with a lower dimensional model is that it contains built-in assumptions which may not be satisfied by the true (unknown) physical state.  Preparing a pure state is strictly speaking rarely achievable due to various experimental imperfections, so using a pure state statistical model is in fact an oversimplification and can lead to erroneous conclusions about the true state. On the other hand, one can argue that when the (small) experimental noises are taken into account, the actual state is ``effectively" low rank, i.e. it has  a small number of significant eigenvalues and a large number of eigenvalues which are so close to zero that they cannot be distinguished from it. Then, the interesting question is how to decide on where to make the cut-off between statistically relevant eigenvalues and pure statistical noise. This is a common problem in statistics which is closely related to that of model selection \cite{Claeskens}. Below we describe the rank-penalised estimator addressing this problem, and show that its theoretical and practical performance is superior to the least squares estimator, and is close to what one would expect from an optimal estimator. In addition, its computation requires only the diagonalisation of the least-squares estimator.

Before presenting a simple algorithm for computing the estimator, we briefly discuss the idea behind its definition. 
Let
\begin{equation}\label{eq.spectral.dec.rho.ls}
\widehat\rho^{(ls)}_n= \sum_{i=1}^d \widehat \lambda_i | \hat{\psi}_i\rangle \langle \hat{\psi}_i |.
\end{equation}
be the spectral decomposition of the least squares estimator, with eigenvalues ordered such that 
$|\widehat \lambda_1| \geq ... \geq |\widehat \lambda_d|$.
For each given rank $\kappa \in \{1,\dots, d\}$ we can project 
$\widehat \rho_n^{(ls)}$ onto the space of matrices of rank $\kappa $ by computing the matrix which is the closest to 
$\widehat \rho_n^{(ls)}$ with respect to the Frobenius distance
$$
\widehat{\rho}_n(\kappa ) := \underset{\tau: \, {\rm rank}(\tau) = \kappa }{\arg\min} \|\tau- \widehat \rho_n ^{(ls)} \|_2^2 .
$$
Although the projection is not a linear operator, $\widehat{\rho}_n(\kappa )$ is easy to compute, and is obtained by truncating the spectral decomposition \eqref{eq.spectral.dec.rho.ls} to the most significant $\kappa $ eigenvalues
$$
\widehat{\rho}_n(\kappa ) = \sum_{i=1}^\kappa  \widehat \lambda_i | \hat{\psi}_i\rangle \langle \hat{\psi}_i |.
$$
The question is now how to choose the rank $\kappa $ in order to obtain a good estimator. In Figure \ref{fig.error.rank} we illustrate the dependence of the norm-two square error $e(\kappa ):= \| \widehat{\rho}_n(\kappa ) -\rho\|_2^2$ on the rank, for a particular 
dataset generated with a  rank 6, 4-ions state. As the rank is increased starting with $\kappa =1$ (pure states), the error 
decreases steeply as $\widehat{\rho}_n(\kappa )$ becomes less biased, it reaches a minimum close to the true rank, and increases slowly as added parameters increase the variance of the estimator. However, since the state $\rho$ is unknown, the norm-two error and optimal rank for which the minimum is achieved, are unknown.   
\begin{figure}
\begin{center}
\includegraphics[width=8cm]{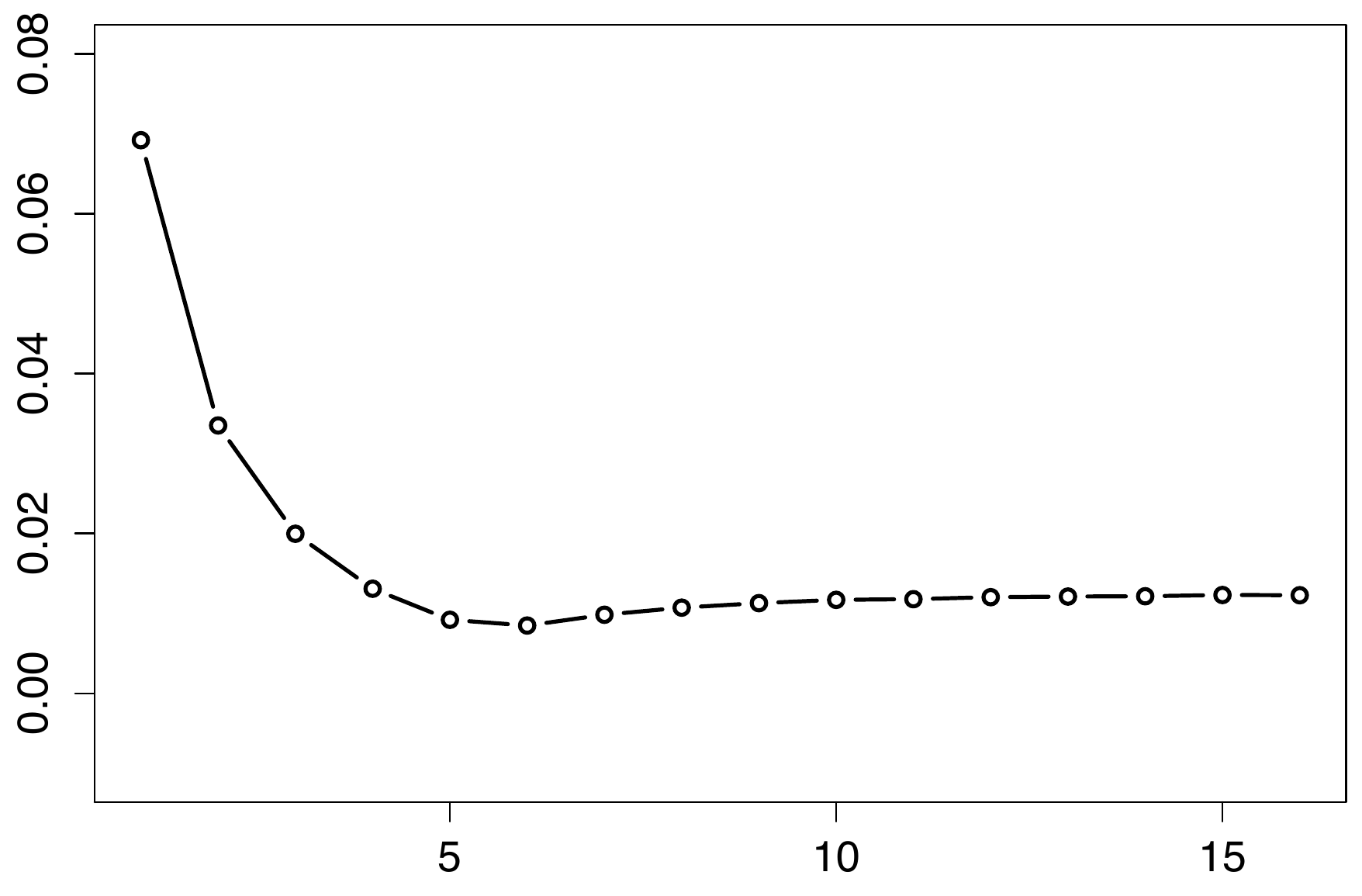}
\end{center}
\caption{The norm-two square error $\| \widehat{\rho}_n(\kappa ) -\rho\|_2^2$ of the truncated LSE as a function of rank, for a rank 6 state, $k=4$ ($d=16$) and $n=500$ measurement repetitions.}
\label{fig.error.rank}
\end{figure}
To go around this, we can estimate the error $e(\kappa )$ from the data by means of e.g. cross-validation, 
as it will be described in section \ref{sec.simulations}. However, in this section we follow a different path, and we define the rank-penalised estimator \cite{AlquierBut2013}, \cite{Bunea} as the minimiser over $\kappa $ of the following expression:
$$
\|\widehat{\rho}_n(\kappa )- \widehat \rho_n ^{(ls)} \|_2^2 + \nu^2 \cdot \kappa   = 
\sum_{i=\kappa +1}^d |\widehat{\lambda}_i|^2 +  \nu^2 \cdot \kappa 
$$
where $\nu$ is a constant which will be tuned appropriately. The first term quantifies the fit of the truncated estimator with respect to the least squares estimator, while the second term is a penalty which increases with the complexity of the model, i.e. the rank. The rank penalised estimator $\widehat{\rho}_n^{(pen)}$ is thus the solution of the simple optimisation problem
\begin{eqnarray}
\label{def:rank-pen}
\widehat{\rho}_n^{(pen)} &:=& \widehat{\rho}_n(\hat{\kappa }), \qquad{\rm where} \qquad
\hat{\kappa }:= \underset{\kappa  = 1,...,d}{\arg\min} \, \left\{ \sum_{i=\kappa +1}^d |\widehat{\lambda}_i|^2 +  \nu^2 \cdot \kappa  \right\}=
\max\{\kappa : \widehat \lambda_{\kappa }^2 \geq \nu^2 \}.
\end{eqnarray}
%
This means that the eigenvalues below a certain noise threshold are set to zero while those above the threshold remain unchanged. The following theorem is our first main result, and shows that the appropriate threshold is given by the upper bound on the operator norm error of the least squares estimator, as established in Proposition~\ref{prop:linear}.

\begin{theorem} \label{thm:nonlinear} 

Let $\theta >0$ be an arbitrary constant, let $c(\theta) := 1+2/\theta$, and let $\varepsilon >0$ be a small parameter. Then 
with probability larger than $1-\varepsilon$, we have
\begin{equation}\label{eq.bound.penalised}
\|\widehat \rho_n ^{(pen)} - \rho\|_2^2 \leq \min_{\kappa =1,...,d} 
\left\{ c^2(\theta) \sum_{j>\kappa } \lambda_j^2(\rho) + 2 c(\theta) \nu(\varepsilon)^2 \kappa 
\right\},
\end{equation}
where $\widehat \rho_n ^{(pen)}$ is the penalised estimator defined in \eqref{def:rank-pen} with threshold 
$\nu( \varepsilon)^2$ given by 
$$
\nu( \varepsilon)^2 :=  \frac 2{n} \left(\frac 23\right)^{k}\log \left(\frac{2^{k+1}}{\varepsilon}\right) = \frac {2d}N \log \frac{2 d}{\varepsilon},
$$
which is assumed to be o(1) with increasing $n$ and $k$.

\end{theorem}

\emph{Proof.} The upper bound follows directly from Proposition \ref{prop:linear} combined with the following oracle inequality established in \cite{AlquierBut2013}, 
$$
\|\widehat \rho_n ^{(pen)} - \rho\|_2^2 \leq \min_{\kappa =1,...,d} 
\left\{ c^2(\theta) \min_{R: \mbox{rank}(R) = \kappa }\|R-\rho\|_2^2 + 2 c(\theta) \cdot \kappa \cdot  \nu(\varepsilon)^2 ,
\right\}
$$
which holds true provided that $(1+\theta) \|\widehat \rho_n^{(ls)} - \rho\|^2 \leq \nu(\varepsilon)^2$. This event occurs with probability larger than $1-\varepsilon$.
\qed

Let us make some explanatory remarks on the above result. Firstly, the bound \eqref{eq.bound.penalised} applies to \emph{all} states $\rho$, not only ``small" rank ones. Recall that $\mathcal{S}_{d,r}$ denotes the set of states of rank-$r$ states on $\mathbb{C}^d$.  In the special case when $\mbox{rank}(\rho) =r \leq d$, the theorem implies that, 
with probability larger than $1-\varepsilon$,
$$
\|\widehat \rho_n ^{(pen)} - \rho\|_2^2 \leq 2 c(\theta)  \cdot r \cdot \nu(\varepsilon)^2 .
$$
If the rank $r$ is much smaller than $d$, this bound is a significant improvement to  the corresponding upper bound $d\cdot \nu(\varepsilon)^2$ for the least square estimator, which can be derived from the operator norm bound of Proposition \ref{prop:linear}. Moreover, up to a constant factor the rate $\nu(\varepsilon)^2$ is equal to
$
d\cdot r \log(d)/N 
$
which is essentially the ratio of the number of parameters and total number of measurements. In section \ref{sec.lower.bound} we will show that apart from the log factor this rate is also optimal, and cannot be improved even if the rank of the state is \emph{known}, which indicates that the estimator adapts to the complexity of the true parameter. Furthermore, we stress the fact that the bound \eqref{eq.bound.penalised} holds true for growing dimension $d=2^k$ as well as the number of measurements $n$; the bound remains meaningful as far as $d \log d /N \to 0$.


The second observation is that our procedure selects the true rank consistently. Denote by $\hat \kappa $ the rank of the resulting estimator $\widehat{\rho}_n^{(pen)}$. Following \cite{AlquierBut2013} we can prove that, if there exists some 
$\kappa$ such that $\lambda_\kappa (\rho) >(1+\delta) \sqrt{\nu (\varepsilon)}$ and  $\lambda_{\kappa +1}(\rho) <(1-\delta) \sqrt{\nu (\varepsilon)}$ for some $\delta \in (0,1)$, then
$$
\mathbb{P}(\hat \kappa  = \kappa ) = 1 - \mathbb{P}(\| \widehat \rho_n^{(ls)} -\rho \| \geq \delta \sqrt{\nu(\varepsilon)}).
$$
This stresses the fact that the procedure detects the eigenvalues above a threshold related to the error of the least squares estimator.
If the true rank of $\rho$ is $r$ and if $k$ and $n$ are such that $\nu$ tends to 0 (which always occurs for fixed number of ions $k$), then $\lambda_r > \nu$ asymptotically and the probability that $\hat \kappa  = r$ tends to 1. 

We can also project $\widehat \rho_n ^{(pen)}$ on the matrices with trace 1, to get
\begin{equation}\label{def:statebis}
\widehat \rho_n ^{(pen,n)} = \arg\min_{R \in \mathcal{S}}
\|R- \widehat \rho_n ^{(pen)} \|_2^2, 
\end{equation}
where $\mathcal{S}$ is the set of all density matrices on $\mathbb{C}^d$. 
The following Corollary shows that the key properties of the estimator 
are preserved if we additionally normalise it to trace-one after thresholding.

\begin{corollary} \label{cor:rk,state} Under the notation and assumptions of Theorem~\ref{thm:nonlinear} 
if $\rho $ is an arbitrary state in $\mathcal{S}_{d,r}$
and if $k$ and $n$ are such that $\lambda_r > \nu(\varepsilon)$ for some $\varepsilon \in (0,1)$, then, with probability larger than $1-\varepsilon$, 
$$
\|\widehat \rho_n ^{(pen,n)} - \rho\|_2^2 \leq 8 c(\theta)\cdot r \cdot \nu(\varepsilon)^2.
$$
Moreover, there exists an absolute constant $C>0$ such that
$$
\sup_{\rho \in \mathcal{S}_{d,r}} E_\rho \|\widehat \rho_n^{(pen,n)} - \rho\|_2^2 
\leq C \frac{rd}N \log\left(\frac{2d}\varepsilon\right).
$$
\end{corollary}
\emph{Proof.} See Appendix~\ref{sec:7.4}. \qed

\subsection{Physical threshold estimator}



Although the rank-penalised estimator performs well in terms of its risk, it is not necessarily positive and trace-one and therefore it may not represent a physical state. In this section we propose and analyse the following ``physical estimator"
\begin{equation}\label{def:state}
\widehat \rho_n ^{(phys)} = \arg\min_{\sigma \in \mathcal{S}(\nu)}
\|\sigma- \widetilde \rho_n ^{(ls)} \|_2^2, 
\end{equation}
where $\widetilde\rho_n^{(ls)}$ is the ```normalised least squares estimator" defined in \eqref{rho.tilde.ls}, and 
$\mathcal{S}(\nu)$ denotes the set of states at noise level $\nu$
$$
\mathcal{S}(\nu) = \left\{ 
\sigma \mbox{ density matrix with eigenvalues } \lambda_j \in \{0\}\cup (4 {\nu},1], \, j=1,...,d
\right\}.
$$
In particular, the space of all density matrices correspond to $\nu=0$ and is denoted $\mathcal{S}$. 
The estimator $\widehat \rho_n ^{(phys)}$ is therefore the physical state which is closest to the (normalised) least square estimator, and whose non-zero eigenvalues are above the threshold $4\nu$.

Before analysing the performance of the estimator, we describe its numerical implementation through the following
simple  iterative algorithm.

Let $ \widetilde \lambda_1\geq \dots \geq   \widetilde \lambda_d$ denote the eigenvalues of 
$\tilde \rho_n^{(ls)}$. Let $\ell=0$, and define $\widehat \lambda_j^{(0)} = \widetilde\lambda_j$ for $j=1,\dots ,d$.

For $\ell = 1,...,d$, do

if $\widehat\lambda_{d-\ell+1}^{(\ell-1)} > {4 \nu}$, STOP;

else, put $\widehat \lambda_{d-\ell+1}^{(\ell)} = 0$ and 
$$
\widehat \lambda_j^{(\ell)} = \widehat \lambda_j^{(\ell-1)} + \frac{1}{d-\ell}\left(1 - \sum_{k=1}^{d-\ell} \widehat \lambda_k^{(\ell-1)} \right), \mbox{ for }j = 1,...,d-\ell;
$$

$\ell= \ell+1$.

The algorithm checks whether the smallest eigenvalue is larger than the noise level $4{\nu}$ and if it is not, then sets its value to 0 and distributes the mass of the erased eigenvalue in such a way that they sum to 1.
This algorithm is similar to that proposed by Smolin et al. \cite{Smolin&Gambetta}, with the important difference that we do not keep all positive eigenvalues but only {\it significantly} positive eigenvalues. Here, significant means larger than the noise threshold of the order of the operator-norm error of the least-squares estimator. Indeed, this noise can give a confidence interval for each eigenvalue.

If the total number of iterations is $\widehat \ell = d-\widehat r$ then the estimator $\widehat \rho_n^{(phys)}$ has rank 
$\widehat r$. Its eigenvalues are equal to 0 for $j > \widehat r$, while, for $j\leq \widehat r$ they are given by 
$$
\widehat \lambda_j^{(phys)} = \widetilde \lambda_j + \frac L2, \qquad \mbox{ where }\qquad
\frac L2 \widehat r = \sum_{ k > \widehat r} \widetilde \lambda_k =  1 - \sum_{k \leq \widehat r} \widetilde \lambda_k .
$$
This implies that $\widehat \rho_n^{(phys)}$ has decreasing eigenvalues and 
$\widehat \lambda_{\widehat r}^{(phys)} \geq 4{\nu}$. The following Theorem shows that $\widehat \rho_n^{(phys)}$ is 
\emph{rank-consistent} and its MSE has the same scaling as that of penalised estimator 
$\widehat{\rho}^{(pen)}_n$.
%
\begin{theorem}\label{thm:recstate}
Assume that the state $\rho$ has rank $r$, i.e. belongs to $\mathcal{S}_{d,r}$. For small $\varepsilon >0$, let $\nu = \nu(\varepsilon)$ be defined as in Theorem~\ref{thm:nonlinear}, and assume that $\lambda_r > 8 \nu(\varepsilon)$. Then, with  $\mathbb{P}_\rho$ probability larger than $1-\varepsilon$ we have $\widehat r = r$
and
$$
\|\widehat \rho_n^{(phys)} - \rho\|_2^2 \leq 48 \cdot r \cdot \nu(\varepsilon)^2.
$$
Moreover, there exists an absolute constant $C>0$ such that
$$
\sup_{\rho \in \mathcal{S}_{d,r}} E_\rho \|\widehat \rho_n^{(phys)} - \rho\|_2^2 
\leq C \frac{rd}N \log\left(\frac{2d}\varepsilon\right).
$$
\end{theorem}
\emph{Proof.} See Appendix \ref{proof.threshold}.


\section{Lower bounds for rank-constrained estimation}\label{sec.lower.bound}

The goal of this section is to investigate how the convergence rates of our estimators compare with that of an 
``optimal estimator" for the statistical model consisting of all states of rank up to $r$. For this we will derive a lower bound for the maximum risk of any estimator.

In this section $\widehat{\rho}_n$ will be an arbitrary estimator and the true state $\rho$ is assumed to belong to the set 
$\mathcal{S}_{d,r}$ of rank-$r$ states. To quantify the overall performance of $\widehat{\rho}_n$, we define the maximum risk  
$$
R_{max}(\widehat{\rho}_n; r) = \sup_{\rho\in \mathcal{S}_{d,r}}  \mathbb{E} _\rho \|\widehat{\rho}_n -\rho\|_2^2. 
$$
In view of the previous upper bounds, we expect its asymptotic behaviour (in terms of the total number of measurements, for a large number of repetitions $n$) to be 
$$
R_{max}(\widehat{\rho}_n;r) = \frac{rd}{N}\cdot  O(1),\qquad N= n\cdot 3^k.  
$$ 
Taking this into account we define the (appropriately rescaled) minimax risk as 
\begin{equation}\label{eq.def.minimax.risk}
R_{minmax}(r,k;n):= \inf_{\widehat{\rho}_n}\,  N R_{max}(\widehat{\rho}_n;r),
\end{equation}
which describes the behaviour of the best estimator at the hardest to estimate state. The next theorem provides an asymptotic lower bound for the minimax risk. It shows that the maximum MSE of any estimator is al least of the order 
of $r(d-r)/N$, which for low rank states scales as $  \# {\rm parameters} /  \# {\rm samples}$, which up to logarithmic factors is the same as the upper bounds derived in Theorems \ref{thm:nonlinear}, Corollary \ref{cor:rk,state}, and Theorem \ref{thm:recstate}. 

\begin{theorem}\label{th.minmax}
The following lower bound holds for the asymptotic minimax risk holds
$$
\underset{n\to \infty}{\lim\inf} \,R_{minmax}(r,k;n)\geq 2r(d-r) .
$$
\end{theorem}

\emph{Proof.}  
The minimax risk captures the worst asymptotic behaviour of the rescaled risk, over \emph{all} states of rank $r$.
In order to bound the risk from \emph{below}, we construct a (lower dimensional) subfamily of states $\mathcal{R}_{d,r}\subset \mathcal{S}_{d,r}$ such that the maximum risk  for this subfamily provides the lower bound. Let 
\begin{equation}\label{eq.rho0}
\rho_0 := {\rm Diag}\left(\frac{1}{r},\dots ,\frac{1}{r}, 0,\dots ,0\right)
\end{equation}
be a diagonal state with respect to the standard basis, and define $\mathcal{R}_{d,r}$ to be the set of matrices 
obtained by rotating $\rho_0$ with an arbitrary unitary $U$, i.e. $\mathcal{R}_{d,r}:= \{\rho:= U \rho_0 U^*\, |\, U \mbox{ unitary}\}$. This is a smooth, compact manifold of dimension $2r(d-r)$ known as a (complex) Grasmannian \cite{}.
%
At each point $\rho = U\rho_0 U^*$ we consider the ONB ${\bf B}_U$ obtained by rotating the standard ONB ${\bf B}$  by the unitary $U$. With respect to this basis, we consider first the parametrisation of an \emph{arbitrary} density matrix 
$\rho^\prime$ by its matrix elements, more precisely by the diagonal, real and imaginary parts of the off-diagonal matrix elements, such that $\rho^\prime\equiv\rho_\theta$ with 
\begin{eqnarray}
\theta&=& (\theta^{(d)}, \theta^{(r)}, \theta^{(i)}) \nonumber\\
&:=& 
(\rho_{11},\dots ,\rho_{dd}; {\rm Re} \rho_{1,2},\dots  ,{\rm Re} \rho_{d-1,d};{\rm Im} \rho_{1,2},\dots,  {\rm Im} \rho_{d-1,d} ).
\label{eq.def.theta}
\end{eqnarray}
The Frobenius distance is given by
$$
\| \rho_{\theta_1} - \rho_{\theta_2} \|^2_2  = 
   \|\theta_1^d- \theta_2^{d}\|^2 +  
2\|\theta_1^r- \theta_2^{r}\|^2+ 
2\|\theta_1^i- \theta_2^{i}\|^2
= (\theta_1- \theta_2)^T G (\theta_1- \theta_2),
$$
where $G$ is the constant diagonal weight matrix 
$
G= {\rm Diag}(\mathbf{1}_{d},\, 2\cdot \mathbf{1}_{d(d-1)/2}, \, 2\cdot \mathbf{1}_{d(d-1)/2}). 
$ 
However, this parametrisation does not take into account the prior information about the rank of the true state, and moreover, our key argument involves the even smaller family $\mathcal{R}_{d,r}$ of states. We will now focus on providing a \emph{local} parametrisation of $\mathcal{R}_{d,r}$ around $\rho = U\rho_0 U^*$. 
With respect to the basis ${\bf B}_U$, a state $\rho^\prime\in \mathcal{R}_{d,r}$ in the neighbourhood of $\rho$ has the form
\begin{equation}\label{eq.decom.rho}
\rho^\prime = 
\rho+ \Delta_{off} + \delta = 
\left(
\begin{array}{ccc}
\frac{1}{r} I_r && 0\\
&&\\
0 &&0
\end{array}
\right)
+ 
 \left(
\begin{array}{ccc}
0 && \Delta \\
&&\\
\Delta^\dagger  && 0
\end{array}
\right)
+
\left(
\begin{array}{ccc}
O(\|\Delta\|^2) && 0\\
&&\\
0  && O(\|\Delta\|^2)
\end{array}
\right).
\end{equation}
where $\Delta$ is a matrix of free (complex) parameters, and the two $O(\|\Delta\|^2)$ blocks are $r\times r$ 
and respectively $(d-r)\times (d-r)$ matrices whose elements scale quadratically in 
$\Delta$ near $\Delta=0$.
The intuition behind this decomposition is that a small rotation of $\rho$ produces off-diagonal blocks of the size of the ``rotation angles" while the change in the diagonal blocks are only quadratic in those angles. Since we are interested in the asymptotic behaviour of estimators, the local approach is justified, and the leading contribution to the Frobenius distance comes from the off-diagonal blocks. More precisely, if $\rho_1,\rho_2\in \mathcal{R}_{d,r}$ are in the neighbourhood of $\rho$ then
$$
\|\rho_1 -\rho_2\|^2 = 2\| \tilde{\theta}_1^r- \tilde{\theta}_2^{r}\|^2+ 
2\|\tilde{\theta}_1^i- \tilde{\theta}_2^{i}\|^2 + O(\| \tilde{\theta}_1\|^4, \| \tilde{\theta}_2\|^4), 
$$
where $\tilde{\theta}^r, \tilde{\theta}^i$ are the real and imaginary parts of the off-diagonal elements contained in 
the block $\Delta$, i.e. for $i\leq r<j$. 
 Locally, the manifold $\mathcal{R}_{d,r}$ can be parametrised by $\tilde{\theta}:= (\tilde{\theta}^r, \tilde{\theta}^i)$.


Since $\mathcal{R}_{d,r}\subset\mathcal{S}_{d,r}$ the maximum risk for the model consisting of rank-$r$ states is bounded from below by that of the (smaller) rotation model $\mathcal{R}_{d,r}$
\begin{eqnarray}\label{r1}
\inf_{\widehat{\rho}_n}\,\sup_{\rho\in \mathcal{S}_{d,r}}  \mathbb{E} _\rho \|\widehat{\rho}_n -\rho\|_2^2 & \geq & \inf_{\widehat{\rho}_n}\,\sup_{\rho\in \mathcal{R}_{d,r}}  \mathbb{E} _\rho \|\widehat{\rho}_n -\rho\|_2^2 
\end{eqnarray}
Let $\pi$ be the ``uniform" distribution over $\mathcal{R}_{d,r}$. To draw a sample from this distribution, one can choose a random unitary $U$ from the Haar measure over unitaries, and defines $\rho := U\rho_0U^*$. Then the maximum risk is bounded from below by the Bayes risk
\begin{equation}\label{r2}
\sup_{\rho\in \mathcal{R}_{d,r}}  \mathbb{E} _\rho \|\widehat{\rho}_n -\rho\|_2^2\geq 
\int_{\mathcal{R}_{d,r}}   \mathbb{E} _\rho \|\widehat{\rho}_n -\rho\|_2^2\, \pi(d\rho)
\end{equation}
By applying the van Trees inequality in \cite{Gill&Levit} (see also \cite{Gill&Massar}) we get that
\begin{equation} \label{r3}
\int_{\mathcal{R}_{d,r}}   \mathbb{E} _\rho \|\widehat{\rho}_n -\rho\|_2^2\, \pi(d\rho)
\geq
\frac 1n \int_{\mathcal{R}_{d,r}} {\rm Tr}( \tilde{G} (\rho)\tilde{I}^{-1}(\rho) ) \pi(d \rho) - \frac \alpha{n^2},
\end{equation}
where $\alpha >0$ is a constant which 
does not depend on $n$. Here, $\tilde{I}(\rho)$ is the (classical) Fisher information matrix of the the data obtained by performing one measurement for each setting, and $\tilde{G}(\rho)$ is the weight matrix corresponding to the quadratic approximation of the Frobenius distance around $\rho$. Both matrices are of dimensions $2r(d-r)= {\rm dim}(\mathcal{R}_{d,r})$, and depend on the chosen parametrisation, but the trace is independent of it.
Inserting (\ref{r2}) and (\ref{r3}) into (\ref{r1}), we get
\begin{equation*}
\inf_{\widehat{\rho}_n}\,\sup_{\rho\in \mathcal{S}_{d,r}} N  \mathbb{E} _\rho \|\widehat{\rho}_n -\rho\|_2^2  \geq  
\frac{N}{n}\int_{\mathcal{R}_{d,r}} {\rm Tr}(  \tilde{G}(\rho)^{1/2}  \tilde{I}(\rho)^{-1}  \tilde{G}(\rho)^{1/2}) \pi(d \rho)  - \frac{\alpha N}{n^2}. 
\end{equation*}
Since $t\mapsto t^{-1}$ is an operator convex function we have
$$
\int  \,  \tilde{G}(\rho)^{1/2}  \tilde{I}(\rho)^{-1}  \tilde{G}(\rho)^{1/2}   \pi(d \rho)  \geq \left( \int   \tilde{G}(\rho)^{-1/2}   \tilde{I} (\rho) \tilde{G}(\rho)^{-1/2}  \pi(d\rho) \right)^{-1}
$$
and by taking the limit $n\to \infty$ we obtain the asymptotic minimax lower bound
\begin{equation}\label{eq.average.risk}
\underset{n\to \infty}{\lim\inf} \,R_{minmax} (r,k)\geq 3^k {\rm Tr} \left( \left( \int   \tilde{G}(\rho)^{-1/2}   \tilde{I} (\rho)  \tilde{G}(\rho)^{-1/2}  \pi(d\rho) \right)^{-1}\right).
\end{equation}
where $R_{minmax} (r,k;n)$ is the minimax risk defined in equation \eqref{eq.def.minimax.risk}


At this point we choose a convenient \emph{local} parametrisation around an arbitrary state $\rho\in \mathcal{R}_{d,r}$. As discussed in the beginning of the proof we showed that for this we can use the real and imaginary parts $\tilde\theta= (\tilde\theta^r,\tilde\theta^i)$ of the off-diagonal block $\Delta$, and that the corresponding weight matrix is $ \tilde{G}(\rho)= 2\mathbf{1}_{2r(d-r)}$. 
The lower bound \eqref{eq.average.risk} becomes 
$$
R_{minmax} (r,k)\geq  3^k\cdot 2\cdot {\rm Tr} \left( \left( \int_{\mathcal{R}_{d,r}}   \tilde{I} (\rho) \pi(d \rho) \right)^{-1}  \right).
$$
Another consequence of \eqref{eq.decom.rho} is that the Fisher information matrix $ \tilde{I} (\rho)$ is equal to the corresponding 
block of the Fisher information matrix $I$ of the full ($d^2$-dimensional) unconstrained model with parametrisation $\theta$ defined in \eqref{eq.def.theta}.
We will now compute the average  over states of the Fisher information with respect to the 
\emph{Pauli bases measurements}, by showing that it is equal to the average  Fisher information at $\rho_0$, for the \emph{random basis measurement}. 
As the different settings are measured independently, the Fisher information $ \tilde{I}(\rho)$ is (and similarly for $I$) 
$$
 \tilde{I}(\rho) =\sum_{{\bf s}\in \mathcal{S}_k }  \tilde{I}(\rho|{\bf s})
$$
where $ \tilde{I}(\rho|{\bf s})$ is the Fisher information corresponding to the von Neumann measurement with respect to the ONB defined by setting ${\bf s}$. More generally, with ${\bf B}_U$ as defined above, we denote by $ \tilde{I}(\rho|{\bf B}_U)$ the Fisher information corresponding to this basis. 
Due to the rotation symmetry, we have
$$
 \tilde{I}(U\rho U^*  \tilde{|}{\bf B}_U) = \tilde{I}(\rho|{\bf B})
$$
so 
\begin{eqnarray*}
\int  \pi(d\rho)   \tilde{I} (\rho) &=& 
\sum _{{\bf s}} \int \pi(d\rho)  \tilde{I}(\rho|{\bf s}) =  
\sum _{{\bf s}} \int \mu^d(dU)  \tilde{I}(U \rho_0 U^* | {\bf s}) \\
\\&=& 3^k  \int  \mu(dU )  \tilde{I}(\rho_0|{\bf B}_U) = 3^k \bar{\tilde{I}}.
\end{eqnarray*}
where  $\mu^d (dU)$ is the unique Haar measure on the unitary group on $\mathbb{C}^d$.

The average Fisher information matrix $\bar{I}$ (of the full, unconstrained model) is computed in section \ref{sec.averageFisher}  where we show that the block corresponding to $\tilde{\theta}$ parameters has average $\bar{\tilde I} = 2 \mathbf{1}_{2r(d-r)}$ such that the lower bound is
%
$$
R_{minmax} (r,k)\geq  2r(d-r).  
$$
%
%
%

\section{Numerical results}\label{sec.simulations}

In this section we present the results of a simulation study which analyses the performance of the proposed estimation methods.  The penalised and physical estimators discussed in the previous sections use a theoretical penalisation and respective threshold rates proportional to $\nu^2$. However in practice we found that the performance of the estimators can be further improved when the rates are adjusted by multiplying with an appropriate constant 
$c$ -- whose choice is informed by the data -- from a grid over a small interval which was chosen to be $[0,3]$. The last two estimators are such versions of the theoretical ones with constant $c$ chosen by using cross-validation methods which are explained in detail in section \ref{sec.procedures}. We will compare the following 5 estimators described below.
\begin{itemize}
 \item[1.] the least squares estimator $\widehat{\rho}_n^{(ls)}$ defined in \eqref{eq.least.squares}.
 \item[2.] the oracle ``estimator" $\widehat{\rho}_n^{(oracle)}$ defined below. This is strictly speaking \emph{not} an estimator since it requires the knowledge of the state $\rho$ itself, and can be computed only in simulation studies. However, the oracle is a useful benchmark for evaluating the performance of the other estimators.   
 \item[3.] the cross-validated projection estimator $\widehat{\rho}_n^{(cv)}$. Here we try to find the optimal truncation rank of the least squares estimator, by using the cross-validation method. 
 \item[4.]the cross-validated penalised estimator $\widehat{\rho}_n^{(pen-cv)}$. This is a modification of the penalised estimator $\widehat{\rho}_n^{(pen)}$ defined in \eqref{def:rank-pen}, where the value of the penalisation constant is adjusted by cross-validation. 
 \item[5.] the cross-validated physical threshold estimator $\rho_n^{(phys-cv)}$. This is a modification of the physical estimator $\widehat{\rho}_n^{(phys)}$ defined in \eqref{def:state}, where the value of the threshold constant is adjusted by cross-validation. 
\end{itemize}
We explore the estimators' behaviour by simulating datasets from states with different ranks, and with different number of measurement repetitions per setting. The methodology is described in detail below.


\subsection{Generation of random states and simulation of datasets}\label{sec.procedures}
In order to generate a density matrix of rank $r$, we first create a rank $r$ upper triangular matrix $T$ in which 
\begin{itemize}
 \item[i)] the off-diagonal elements of the first $r$ rows are random complex numbers,
 \item[ii)] the diagonal elements $T_{22}, \dots , T_{rr} $ are real, positive random numbers,
 \item[iii)] all elements of the rows $r+1,\dots, d$ are zero.
 \end{itemize}
The matrix $T$ is completed by setting $T_{11}$ such that $T_{11}\geq 0$, and $\|T\|_2^2=1$. If these conditions cannot be satisfied we repeat the procedure by generating a new set of matrix elements for $T$. When successful, we set 
$\rho := T^{*}T$ which by construction is a density matrix of rank $r$. We note that it is not our purpose to generate matrices from a particular ``uniform" ensemble, but merely to have a state with reasonably random eigenvectors, and 
whose $r$ eigenvalues are not significantly smaller than $1/r$. Following this procedure we have generated 4 states of 4 ions ($d= 2^4$) with ranks $1,2,6,10$. The rank $6$ state for instance, has non-zero eigenvalues $(0.47, 0.19, 0.12, 0.11, 0.07, 0.04)$.

For each state, we have then simulated a number of 100 independent datasets with a given number of repetitions chosen from the range $20, 100, 500$, and $2500$. In this way we can study the dependence of the MSE of each estimator on state (or rank) and number of repetitions.

%

\subsection{Computation of estimators}
We conducted the following simulation study for all the possible combinations between the states and the total number of cycles (i.e. $4 \times 4 = 16 $ different scenarios). Below, we denote by $r$ the rank of the ``true" state $\rho$, from which the data has been generated. The procedure has the following steps:


\begin{enumerate}

\item For a given number of repetitions $n$, we simulate $5$ independent datasets $D_1, \dots, D_5$, each with $n/5$ repetitions. By simply adding the number of counts for each setting and outcome, we obtain a dataset $D$ of 
$n$ repetitions. However, as we will see below, having $5$ separate ``smaller" datasets is important for the purpose of applying cross-validation. Note that such a procedure can be easily implemented in an experimental setting.
%
%

 \item We compute the \emph{least square estimator} $\widehat{\rho}_n^{(ls)}$ based on the full dataset $D$ with total number of cycles $n$.

 \item We compute the \emph{oracle ``estimator"} as follows:
 \begin{enumerate}
\item 
we compute the spectral decomposition \eqref{eq.spectral.dec.rho.ls} of $\widehat{\rho}_n^{(ls)}$, with the eigenvalues $\widehat{\lambda}_i$ arranged in decreasing order of their absolute values. For each rank $1\leq \kappa \leq d$ we define the truncated (least squares) matrix
$$
\widehat{\rho}_n(\kappa) = \sum_{i=1}^\kappa \widehat \lambda_i | \hat{\psi}_i\rangle \langle \hat{\psi}_i |.
$$
 
\item 
we then evaluate the norm-two (Frobenius) distance $e(\kappa) := \|\rho - \widehat{\rho}_n(\kappa) \|_2^2$ and define the oracle estimator 
as the truncated estimator with minimal norm two error 
   $$
   \widehat{\rho}^{(oracle)}_n =  \widehat{\rho}_n(\kappa_0), \qquad \kappa_0 =\arg\min_\kappa e(\kappa).
    $$
    Note that the oracle estimator relies on the knowledge of the true state $\rho$ which is not available in a real data set-up. It is nevertheless useful as a benchmark for judging the performance of other estimators in simulation studies. At the next point we define the cross-validation estimator which tries to find the ``optimal" rank $\kappa_0$ by replacing the unknown state $\rho$ with the least squares estimator computed on a separate batch of data.  
 \end{enumerate}

\item We compute the \emph{cross validation estimator} as follows.
\begin{enumerate}
\item For each $ j\in \{1,\dots ,5 \}$ we compute the following estimators. While holding the batch $D_j$ out, we compute the least squares estimator $\widehat{\rho}_{n;-j}^{(ls)}$ for the dataset consisting of joining the remaining 4 batches together. Similarly to the point above, we define the rank $\kappa$ truncation of this estimator by 
$\widehat{\rho}_{n;-j}(\kappa)$. We also compute the least squares estimator for the remaining batch $j$, denoted by 
$\widehat{\rho}_{n;j}^{(ls)}$.


\item For each rank $\kappa$ we evaluate the ``empirical discrepancy"
$$
\mbox{CV}(\kappa) = \frac{1}{5}\sum_{i=1}^{5}\left\|\widehat{\rho}_{n;-j}(\kappa) - \widehat{\rho}_{n;j}^{(ls)}\right\|_2^2.
$$
Since $\widehat{\rho}_{n;-j}(\kappa)$ and $\widehat{\rho}_{n;j}^{(ls)}$ are independent, and the least squares estimator is unbiased $\mathbb{E}(\widehat{\rho}_{n;j}^{(ls)}) = \rho$, the expected value of $\mbox{CV}(k)$ is
\begin{eqnarray*}
\mathbb{E}\left[\mbox{CV}(\kappa)\right] & = &
\mathbb{E}_{-1} \left[  {\rm Tr} \left(\widehat{\rho}_{n;-1}(\kappa)^2\right)\right] - 
2 {\rm Tr} \left(  \mathbb{E}_{-1}\left[   \widehat{\rho}_{n;-1}(\kappa) \right]\cdot \mathbb{E}_{1}\left[  \widehat{\rho}_{n;1}^{(ls)} \right]     \right) +
\mathbb{E}_1\left[  {\rm Tr} \left(\left(\widehat{\rho}^{(ls)}_{n;1}\right)^2\right)\right] \\
&=& 
\mathbb{E}_{-1} \left[  {\rm Tr} \left(\widehat{\rho}_{n;-1}(\kappa)^2\right)\right] - 
2 {\rm Tr} \left(  \mathbb{E}_{-1}\left[   \widehat{\rho}_{n;-1}(\kappa) \right] \rho \right) + {\rm Tr}(\rho^2) + C\\
&=& \mathbb{E}_{-1} \left[  \| \widehat{\rho}_{n;-1}(\kappa) - \rho \|^2 \right] +C
\end{eqnarray*}
where we denoted $\mathbb{E}_{-1}$ and $\mathbb{E}_{1}$ the expectation over all batches except the first, and respectively over the first batch. Therefore, the average of $\mbox{CV}(\kappa)$ is equal to the mean square error of the truncated estimator $\widehat{\rho}_{n;-1}(\kappa)$, up to a constant $C$ which is independent of $\kappa$.

\item Based on the above observation we use the the cross-validation method as a proxy for the oracle estimator. Concretely, we minimize \mbox{CV}$(\kappa)$ with respect to $\kappa$
$$
\hat{\kappa}_{cv} = \underset{k}{\arg\min} \mbox{ CV}(\kappa),
$$
and define the cross-validation estimator as the truncation to rank $\hat{\kappa}_{cv} $ 
of the \emph{full data} least squares estimator 
$\widehat{\rho}_n^{(cv)} := \widehat{\rho}_n(\hat{\kappa}_{cv})$.
\end{enumerate}

\item We compute the \emph{cross-validated rank-penalised estimator} as follows.
\begin{enumerate}
  \item Let $c$ be a penalisation constant chosen from a suitable set of discrete values in the interval $[0,3]$. Similarly to the cross-validation procedure, we hold out batch $j$, and we compute the rank-penalised estimator \eqref{def:rank-pen}, with penalty constant $c\nu^2$ for $j=1,\ldots, 5$. We denote these estimators by $\widehat{\rho}_{n;-j}^{(pen)} (c)$. We will also need the least square estimator $\rho^{(ls)}_{n;j}$ 
  for batch $j$ computed above.
  \item For each value of $c$ we evaluate the empirical discrepancy
  $$
  \mbox{CV}(c) = \frac{1}{5}\sum_{i=1}^{5}\left\|\widehat{\rho}_{n;-j}^{(pen)} (c) - \rho^{(ls)}_{n;j}\right\|_2^2
  $$
 and minimize $\mbox{CV}(c)$ with respect to the constant $c$
 $$
 \widehat{c} = \underset{c}{\arg\min} \mbox{ CV}(c).
 $$
Finally we compute the cross-validated rank penalised estimator $\widehat{\rho}^{(rk-cv)}_n$ which is defined as in \eqref{def:rank-pen}, with constant $\widehat{c}\nu^2$, on the whole dataset $D$. 
\end{enumerate}

\item We compute the \emph{physical estimator} as follows.
\begin{enumerate}
  \item As above we choose a constant $c$ from a grid over the interval $ [0,3]$. We hold out batch $j$, and we compute the physical threshold estimator \eqref{def:state}, using the algorithm below this equation, with threshold $c\cdot 4 \nu$. 
  We denote the resulting estimators by $\widehat{\rho}^{(phys)}_{n;-j}(c)$, for $j=1,\ldots, 5$.
  \item For each value of $c$ we evaluate the empirical discrepancy
  $$
  \mbox{CV}(c) = \frac{1}{5}\sum_{i=1}^{5}\left\| \widehat{\rho}^{(phys)}_{n;-j}(c) -\widehat{\rho}^{(ls)}_{n;j}\right\|_2^2.
  $$
We then minimize $\mbox{CV}(c)$ with respect to the constant $c$. 
$$
\hat{c} = \underset{c}{\arg\min} \mbox{ CV}(c)
$$
Finally we compute the cross-validated physical estimator $\widehat{\rho}^{(phys-cv)}_n$ which is defined as in \eqref{def:state}, with constant $\hat{c}\cdot 4\nu$.
\end{enumerate}

\end{enumerate}

\subsection{Simulation results}\label{sec.simulation.results}
We collect here the results of the simulation study described in the previous section. As a figure of merit we focus on the mean square error $\mathbb{E}\left((\|\widehat{\rho}_n- \rho\|_2^2\right)$ of each estimator, which is estimated by averaging  the square errors over the 100 independent repetitions of the procedure. We are also interested in how the different methods perform relative to each other, and whether the selected rank is consistent, i.e. it concentrates on the rank of the true state for large number of repetitions.  

\begin{figure}[H]
\begin{center}
\begin{minipage}[b]{0.45\linewidth}
\centering
\includegraphics[width=.9\linewidth]{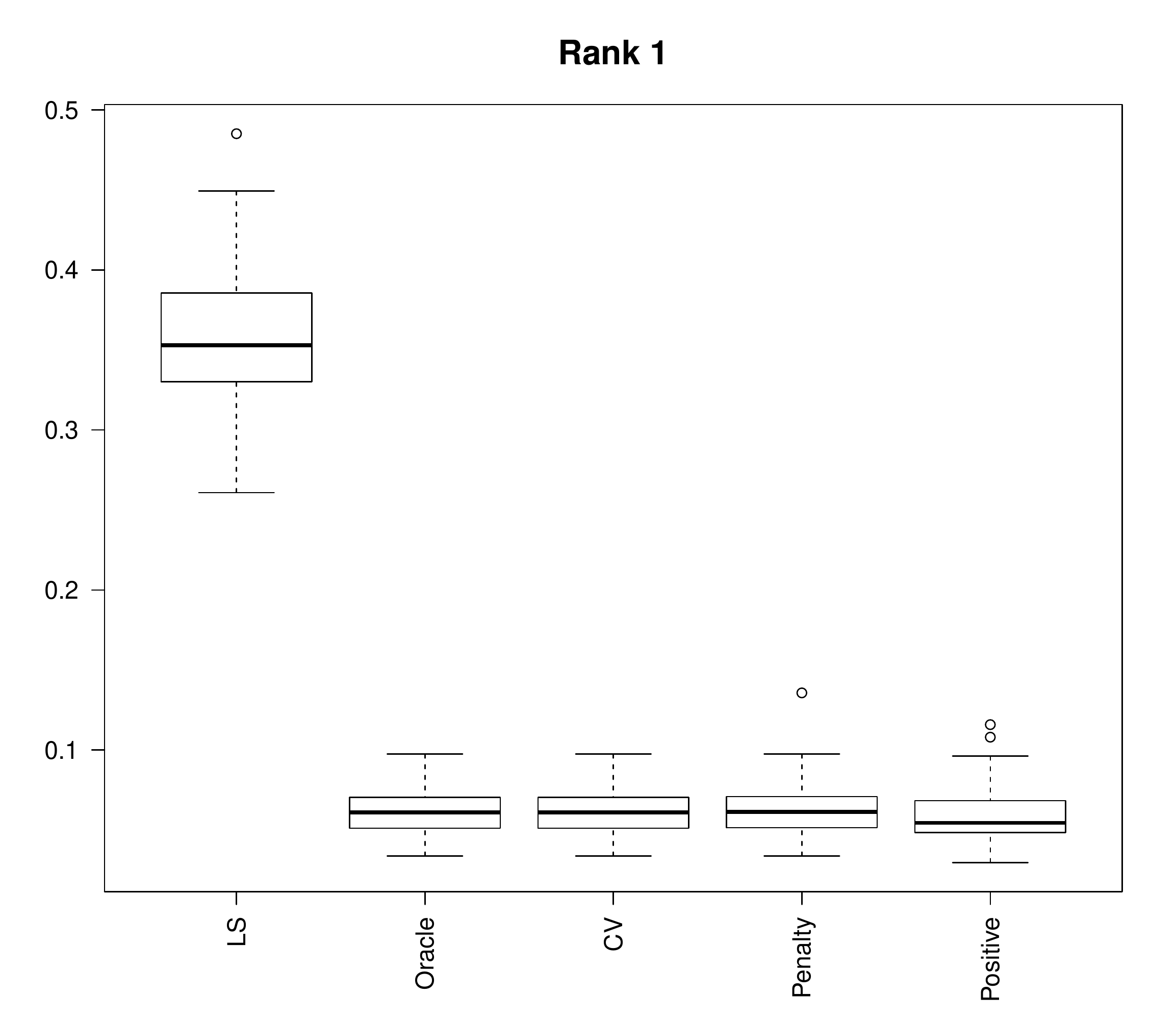}\\
a) MSEs for a rank 1 state  
 \end{minipage}
 \quad
 \begin{minipage}[b]{0.45\linewidth}
 \centering
\includegraphics[width=.9\linewidth]{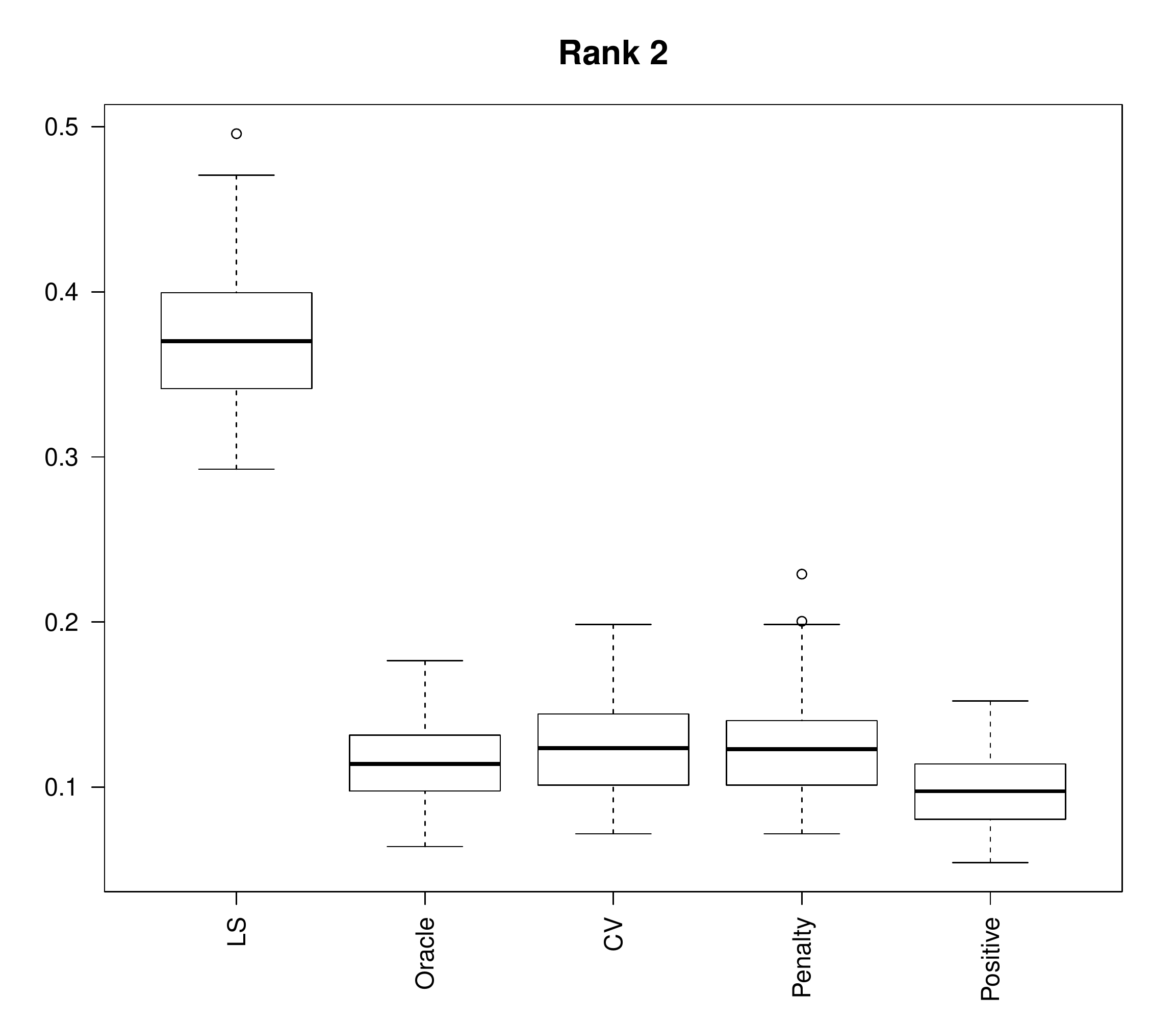}\\
b) MSEs for a rank 2 state 
\end{minipage}

\vspace{4mm}

\begin{minipage}[b]{0.45\linewidth}
\centering
\includegraphics[width=.9\linewidth]{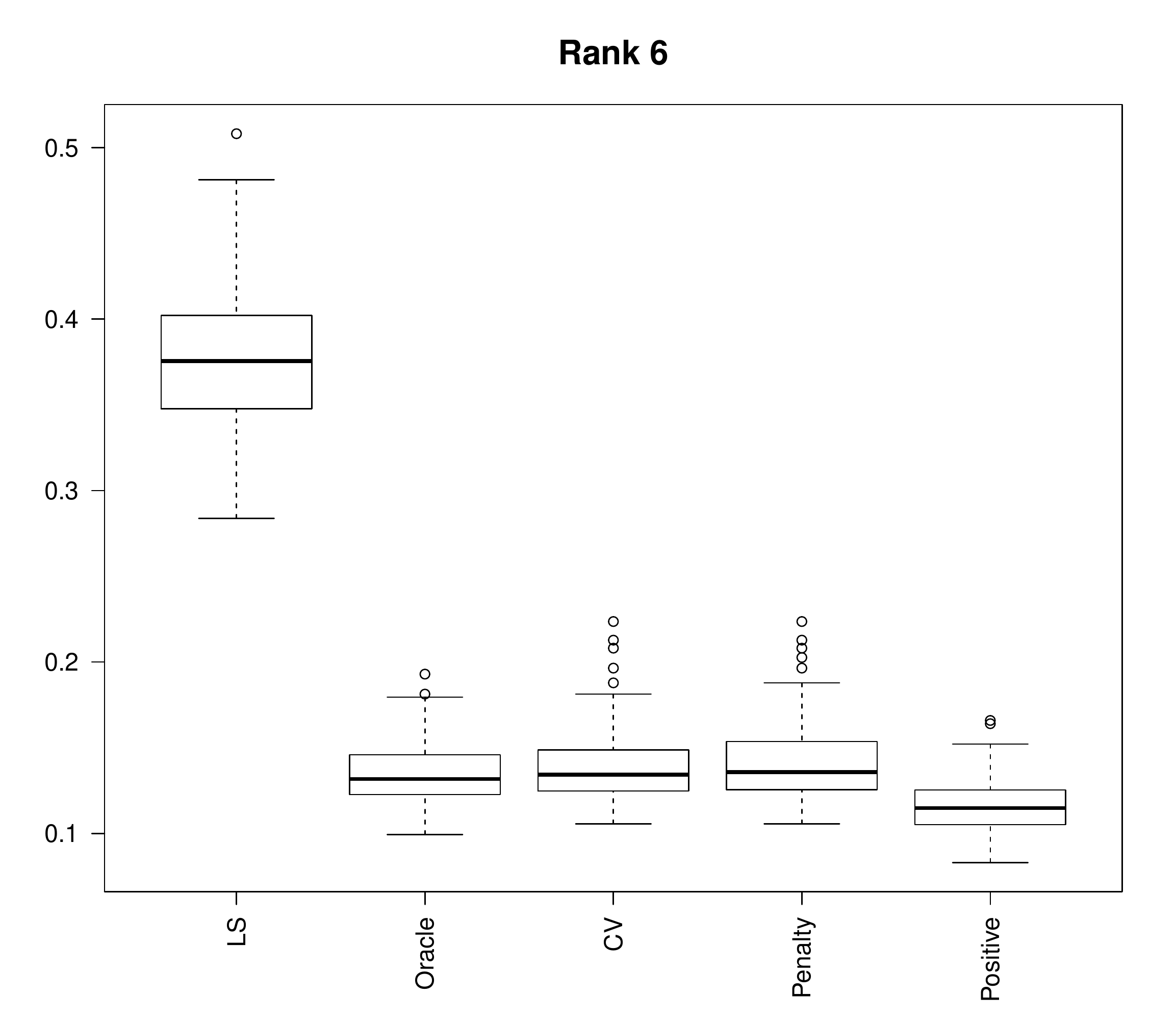}\\
c) MSEs for a rank 6 state  
 \end{minipage}
 \quad
 \begin{minipage}[b]{0.45\linewidth}
 \centering
\includegraphics[width=.9\linewidth]{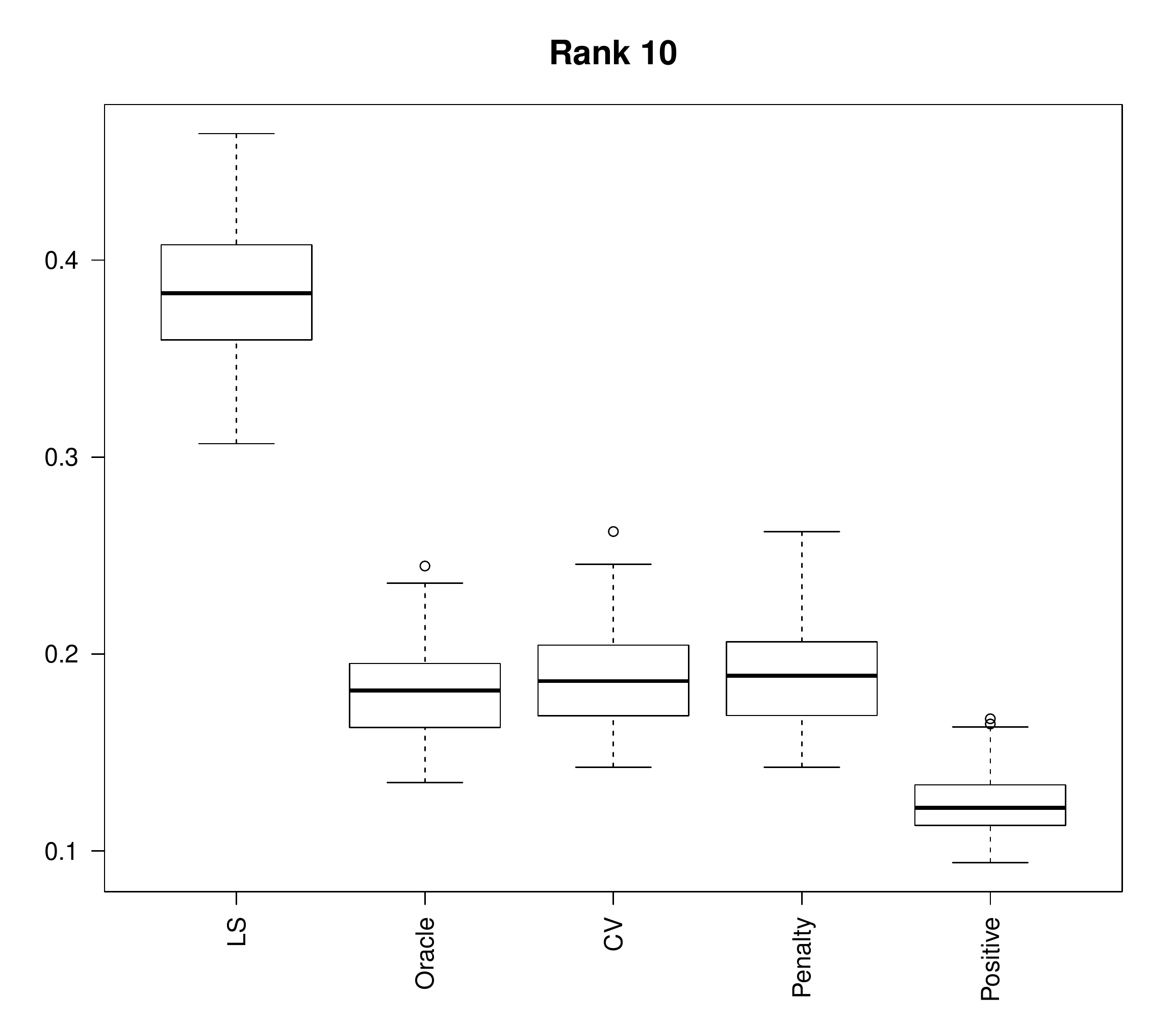}\\
d) MSEs for a rank 10 state 
\end{minipage}
\caption{Boxplots for the estimated mean square error ($\mathbb{E} \|\widehat{\rho}_n - \rho\|_2^2$)} for different 
ranks, $k=4$ ($d=16$), with $n=20$ repetitions
\label{fig:MSE_20}
\end{center}
\end{figure}
The four panels in Figure \ref{fig:MSE_20} represent the boxplots of the square errors $\|\widehat{\rho}_n- \rho\|_2^2$ for the different estimators, and different states, when the number of repetitions is $n=20$. Similarly,  Figure \ref{fig:MSE_100} shows the same boxplots at $n=100$. As expected, in both cases the least squares performs significantly worse than the other estimators, and the discrepancy is larger for small rank states. The remaining 4 estimators have similar MSE's with the physical one performing slightly better than the rest, followed by the oracle. 
Note also that the estimators' variances (indicated by the size of the boxes) are larger for the least squares than the other estimators. A similar behaviour has been observed for $n=500, 2500$ repetitions.

\begin{figure}[H]
\begin{center}
\begin{minipage}[b]{0.45\linewidth}
\centering
\includegraphics[width=.9\linewidth]{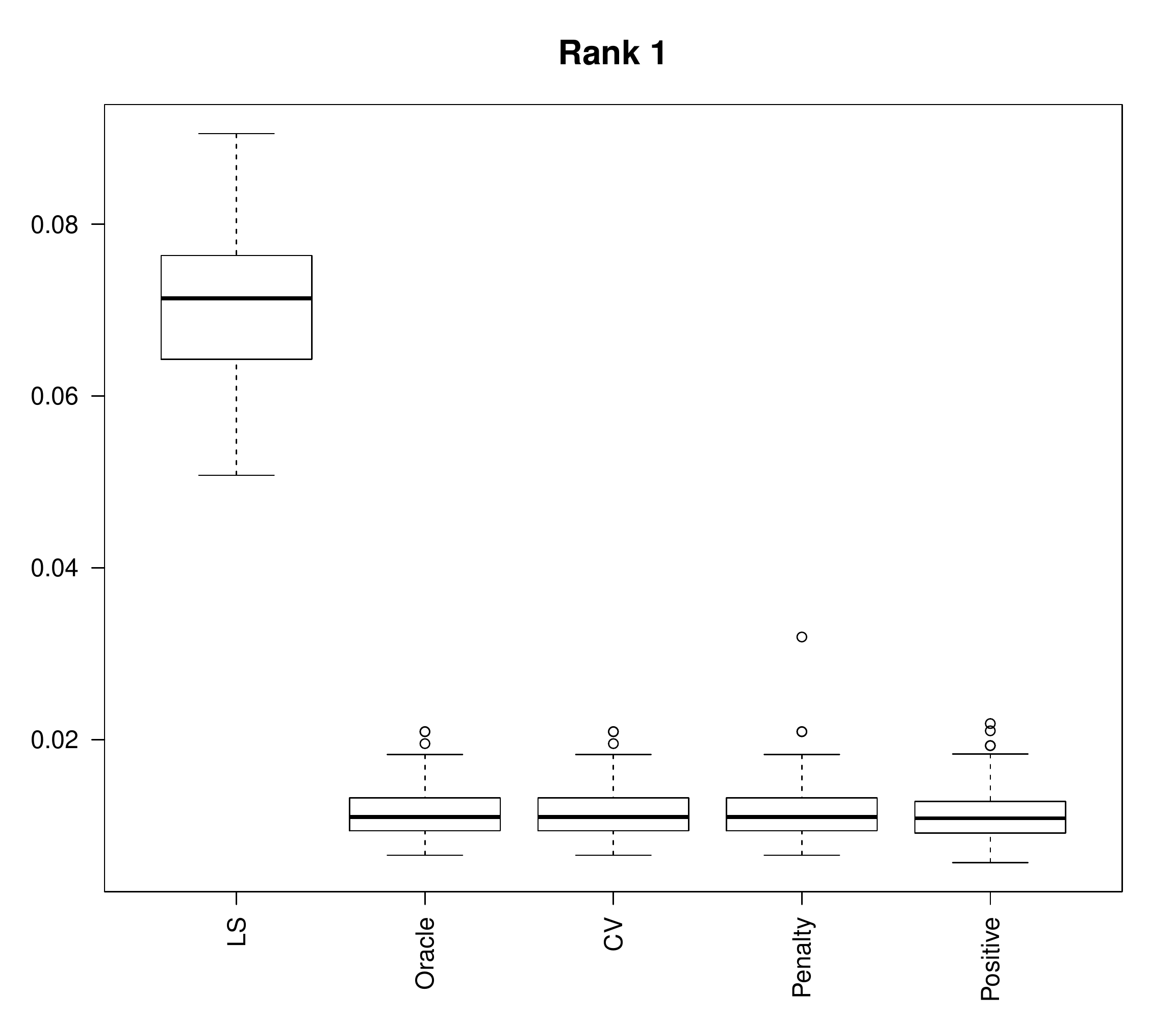}\\
a) MSEs for a rank 1 state  
 \end{minipage}
 \quad
 \begin{minipage}[b]{0.45\linewidth}
 \centering
\includegraphics[width=.9\linewidth]{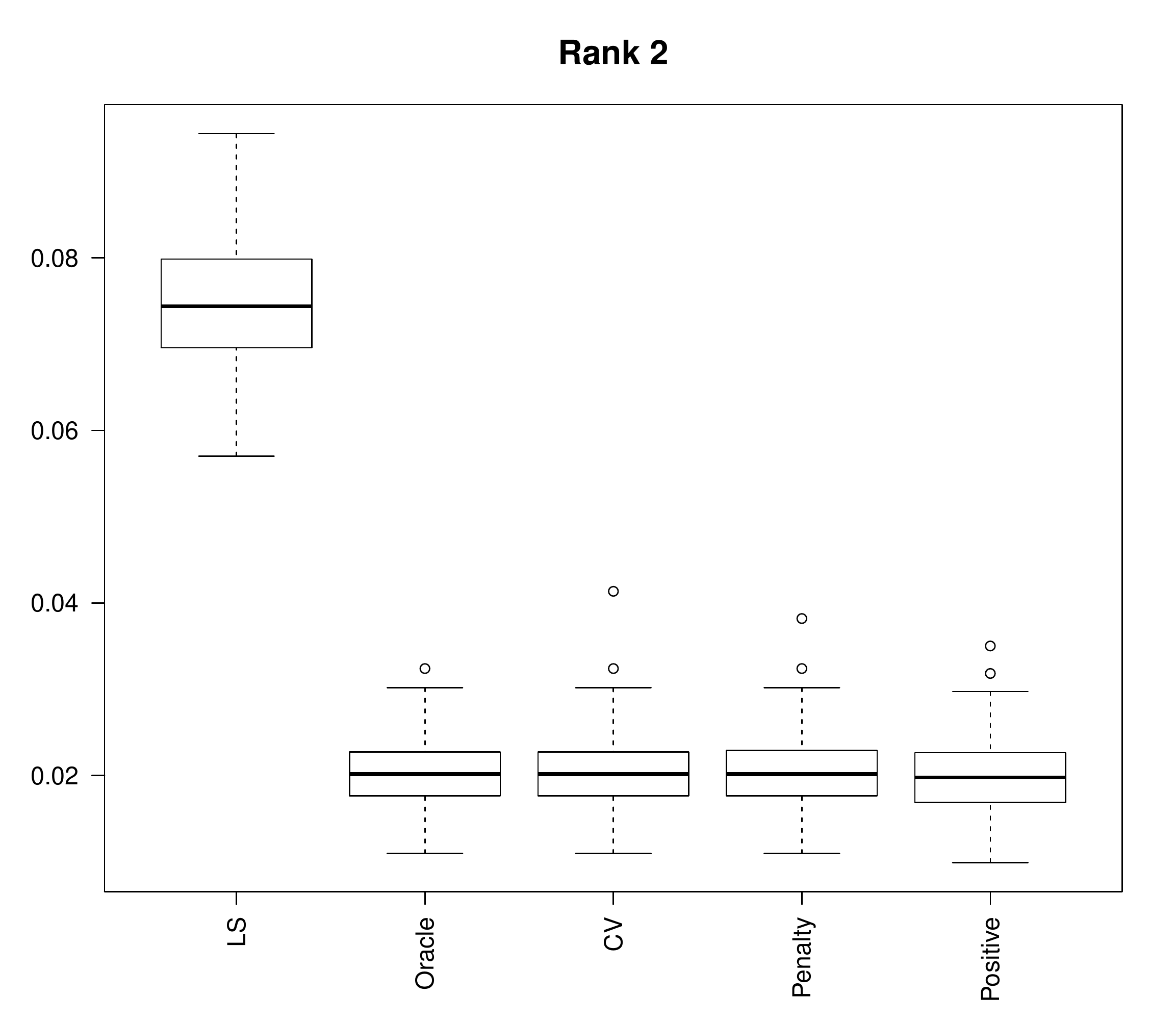}\\
b) MSEs for a rank 2 state 
\end{minipage}

\vspace{4mm}

\begin{minipage}[b]{0.45\linewidth}
\centering
\includegraphics[width=.9\linewidth]{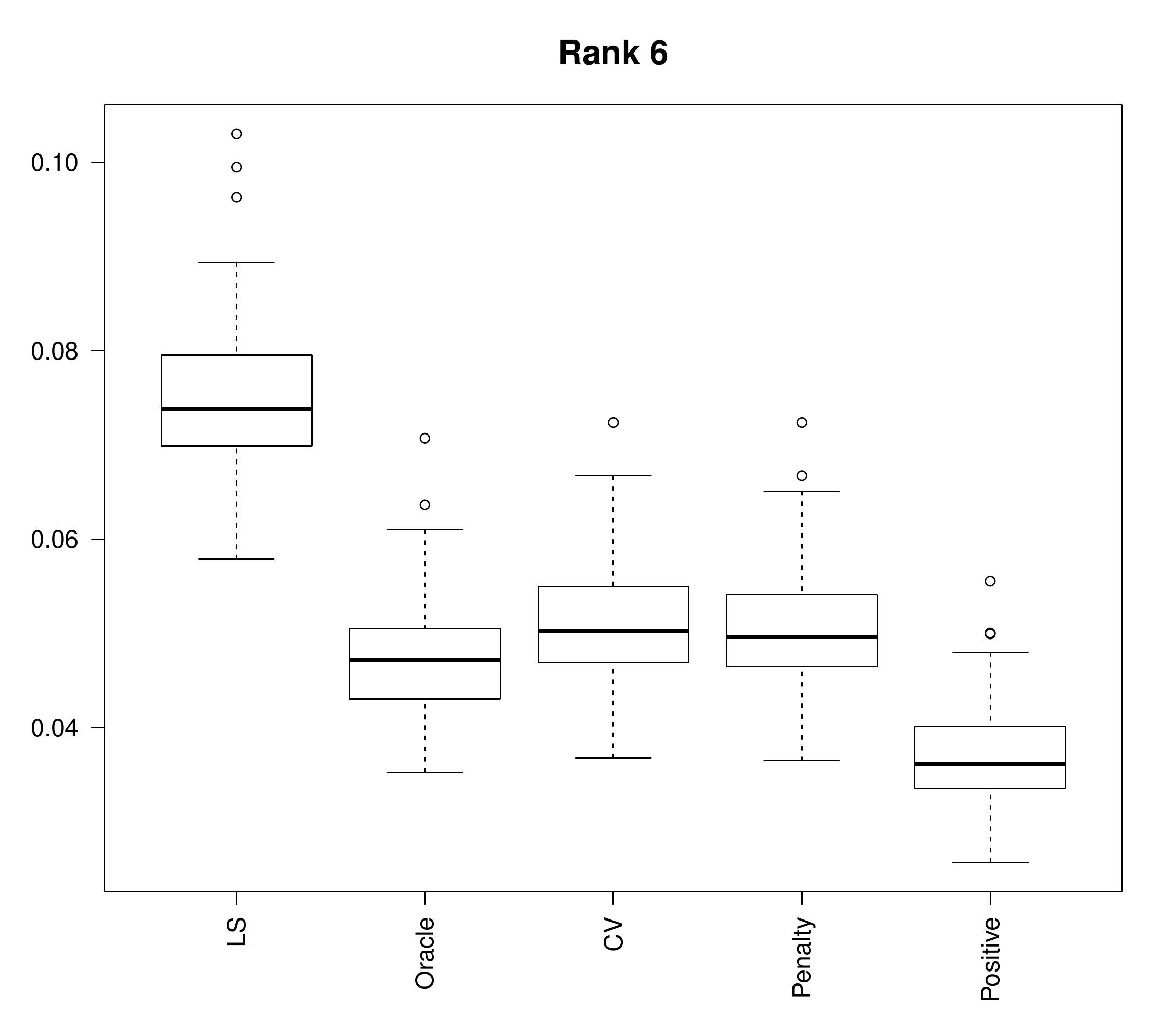}\\
c) MSEs for a rank 6 state  
 \end{minipage}
 \quad
 \begin{minipage}[b]{0.45\linewidth}
 \centering
\includegraphics[width=.9\linewidth]{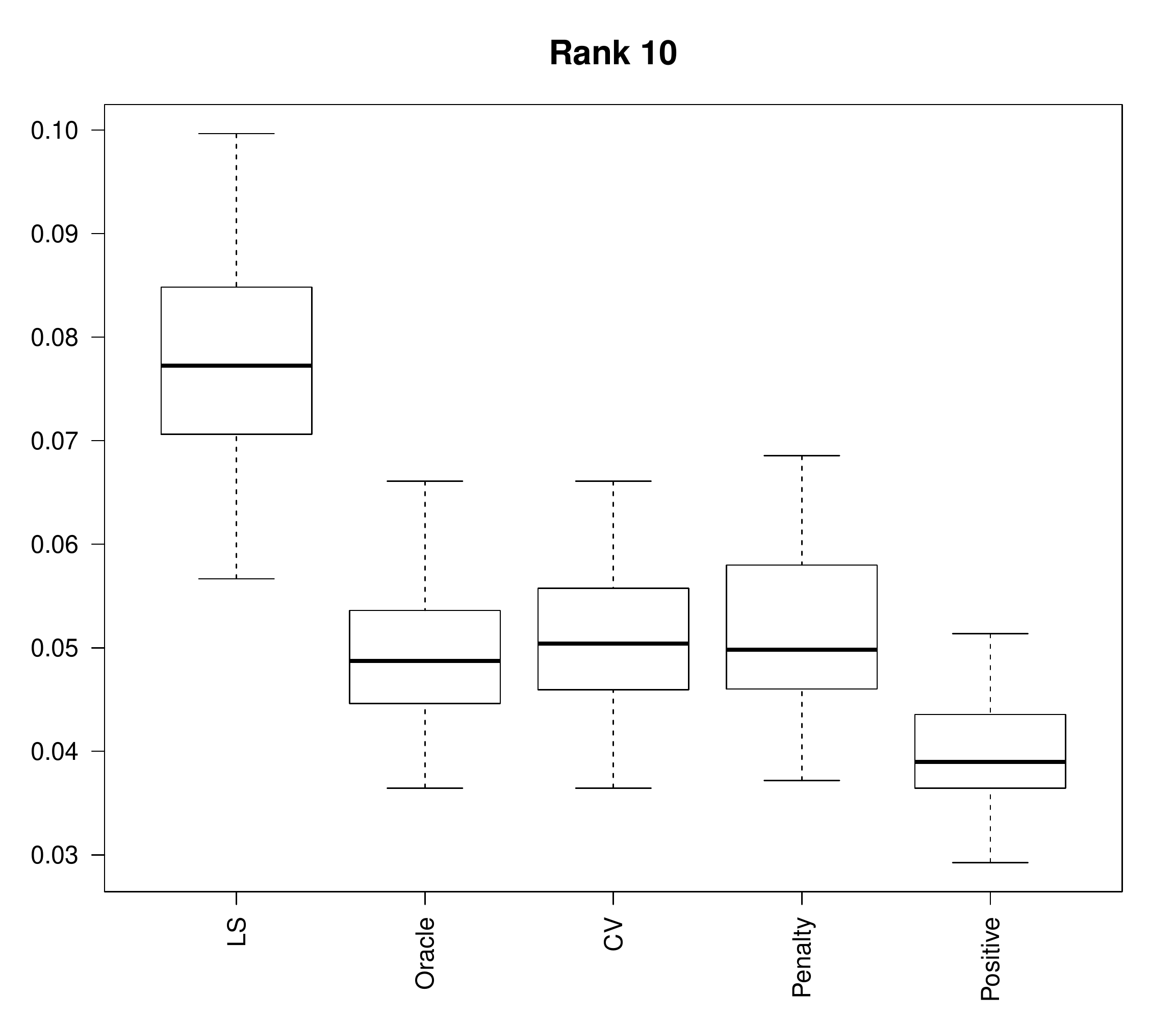}\\
d) MSEs for a rank 10 state 
\end{minipage}
\caption{Boxplots for the estimated mean square error ($\mathbb{E} \|\widehat{\rho}_n - \rho\|_2^2$)} for different 
ranks, $k=4$ ($d=16$), with $n=100$ repetitions
\label{fig:MSE_100}
\end{center}
\end{figure}
Figure \ref{fig:mse.with.n} illustrates the dependence of the MSE of a given estimator, as a function of $n$, for the four different states which have been analysed. Since the MSE decreases as $n^{-1}$ we have chosen to plot the ``renormalised" MSE given by $ n \cdot \mathbb{E} \|\widehat{\rho}_n- \rho\|_2^2$, which converges to a constant value for large $n$. As expected the limiting value increases with the rank of the state, as a proxy for the number of parameters to be estimated.
\begin{figure}[H]
\begin{center}
\begin{minipage}[b]{0.45\linewidth}
\centering
\includegraphics[width=.9\linewidth]{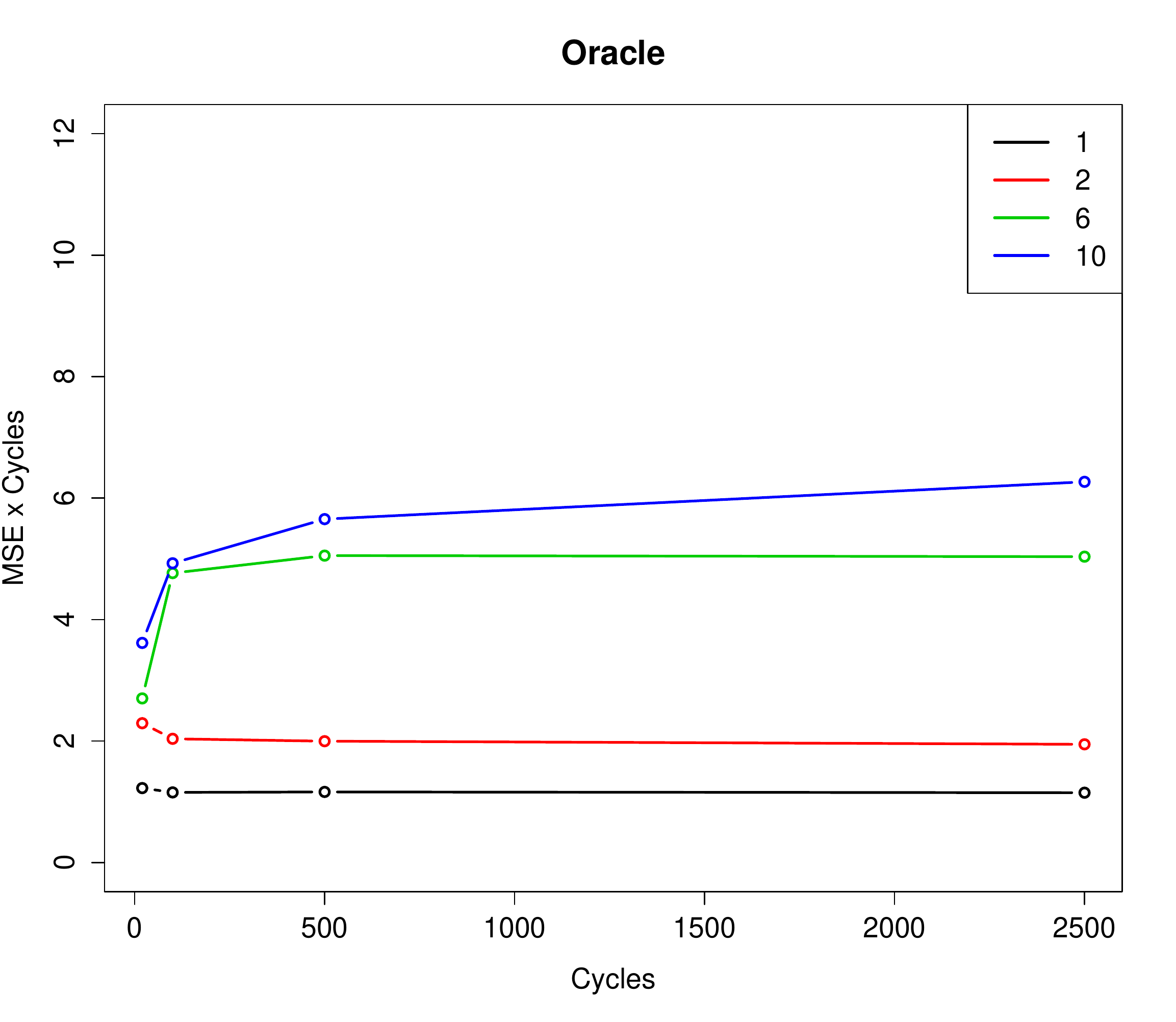}\\
a) ``Renormalised" MSEs for oracle estimator  
 \end{minipage}
 \quad
 \begin{minipage}[b]{0.45\linewidth}
 \centering
\includegraphics[width=.9\linewidth]{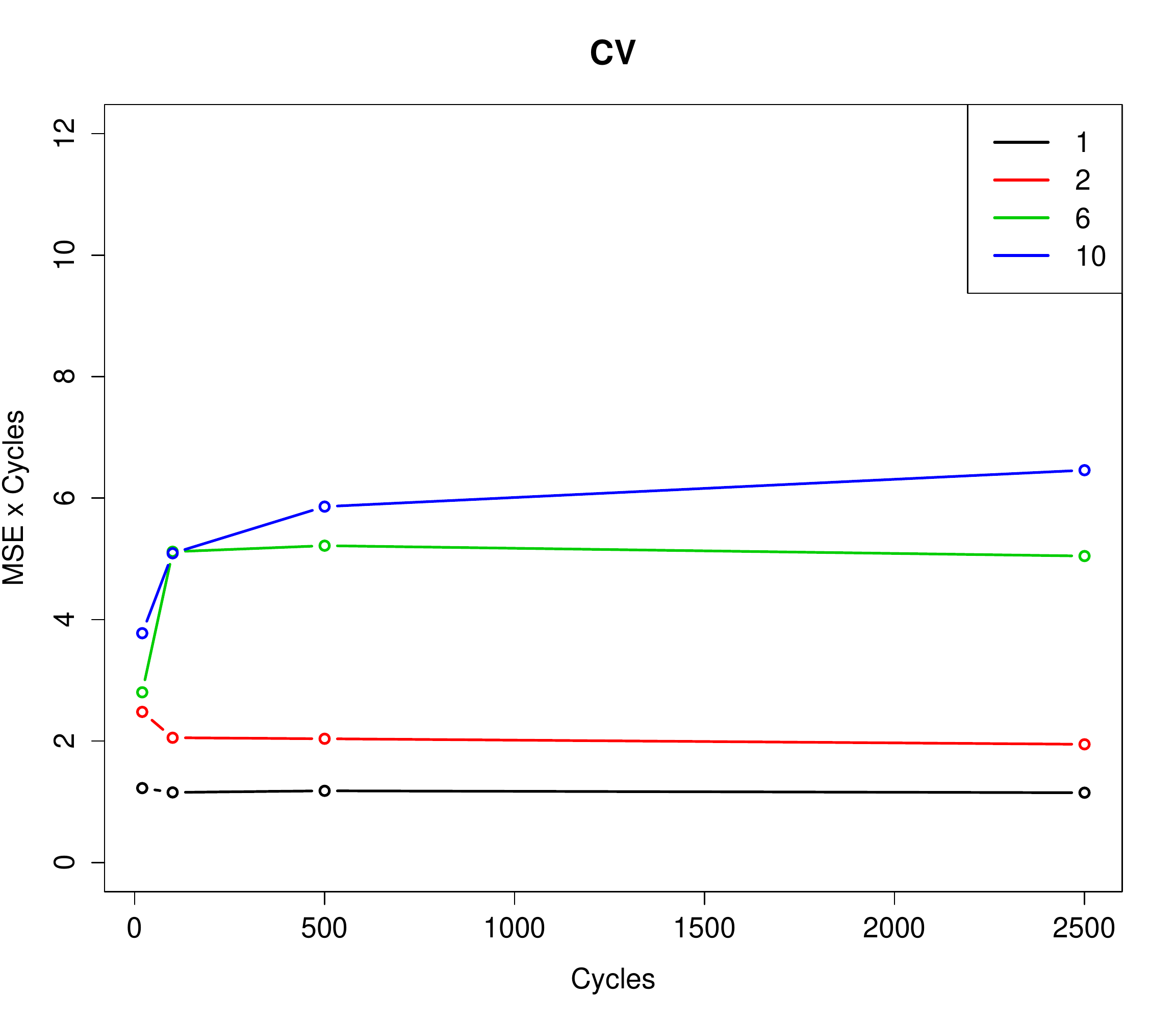}\\
b) ``Renormalised" MSEs for cross-validation estimator
\end{minipage}

\vspace{4mm}

\begin{minipage}[b]{0.45\linewidth}
\centering
\includegraphics[width=.9\linewidth]{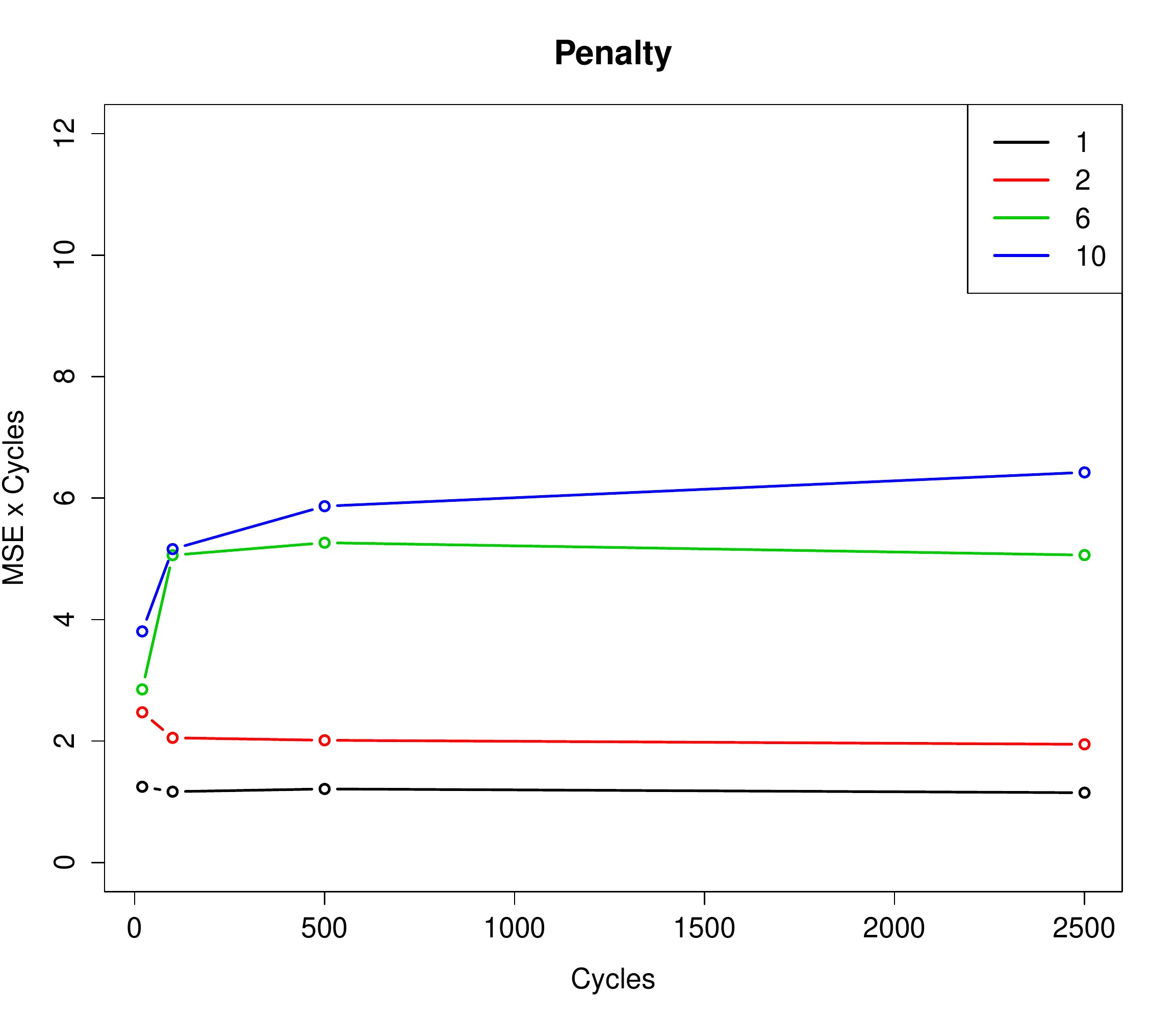}\\
c) ``Renormalised" MSEs for penalised estimator
 \end{minipage}
 \quad
 \begin{minipage}[b]{0.45\linewidth}
 \centering
\includegraphics[width=.9\linewidth]{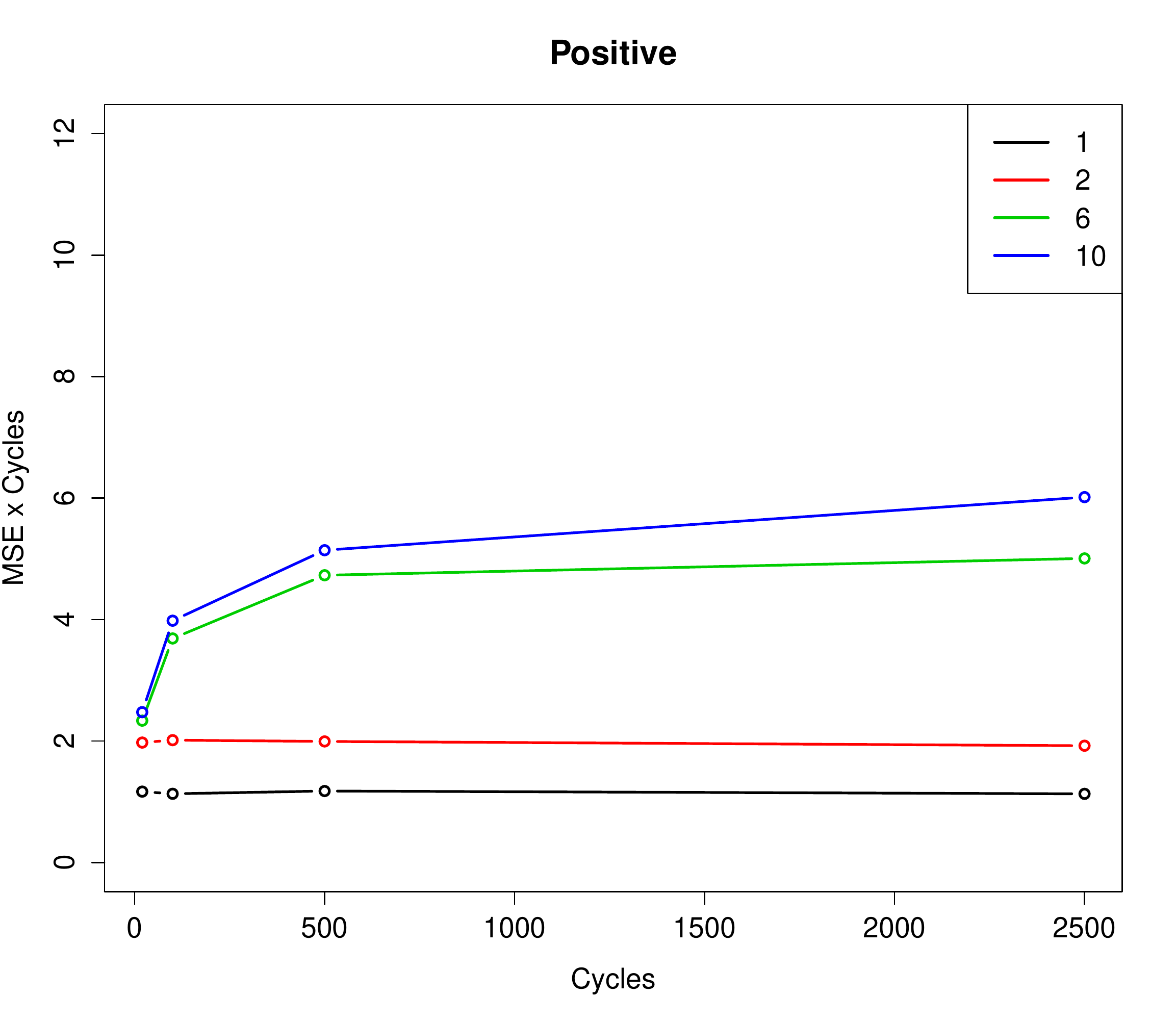}\\
d) ``Renormalised" MSEs for positive estimator 
\end{minipage}
\caption{Renormalised MSEs $n \cdot \mathbb{E} \|  \widehat{\rho}_n -\rho\|_2^2$ as a function of the number of repetitions  for states with different ranks: 1(black), 2 (red), 6 (green), 10 (blue) }.
\label{fig:mse.with.n}
\end{center}
\end{figure}

The histograms in Figure \ref{fig:histogram.chosen.rank} show the probability distributions for the chosen rank of each given estimator, as a function of the number of measurement repetitions $n$, for the state of rank $6$. 
We note that in all cases the proportion of times that the chosen rank is equal to the true rank of the state increases as with the number of repetitions. However, this convergence towards a ``rank-consistent" estimator is rather slow, as the proportion surpasses $80\%$ only when $n=2500$. Another observation is that the penalty and threshold estimators appear to have different behaviours: the former tends to underestimate the true rank, while the latter tends to overestimate it. As expected, the oracle estimator is more likely to choose the correct rank for large number of repetitions. Perhaps slightly more surprising, for small number of repetitions ($n=20$), the oracle choose a pure state in most cases.
\begin{figure}[t!]
\begin{center}
\begin{minipage}[b]{0.45\linewidth}
\centering
\includegraphics[width=.9\linewidth]{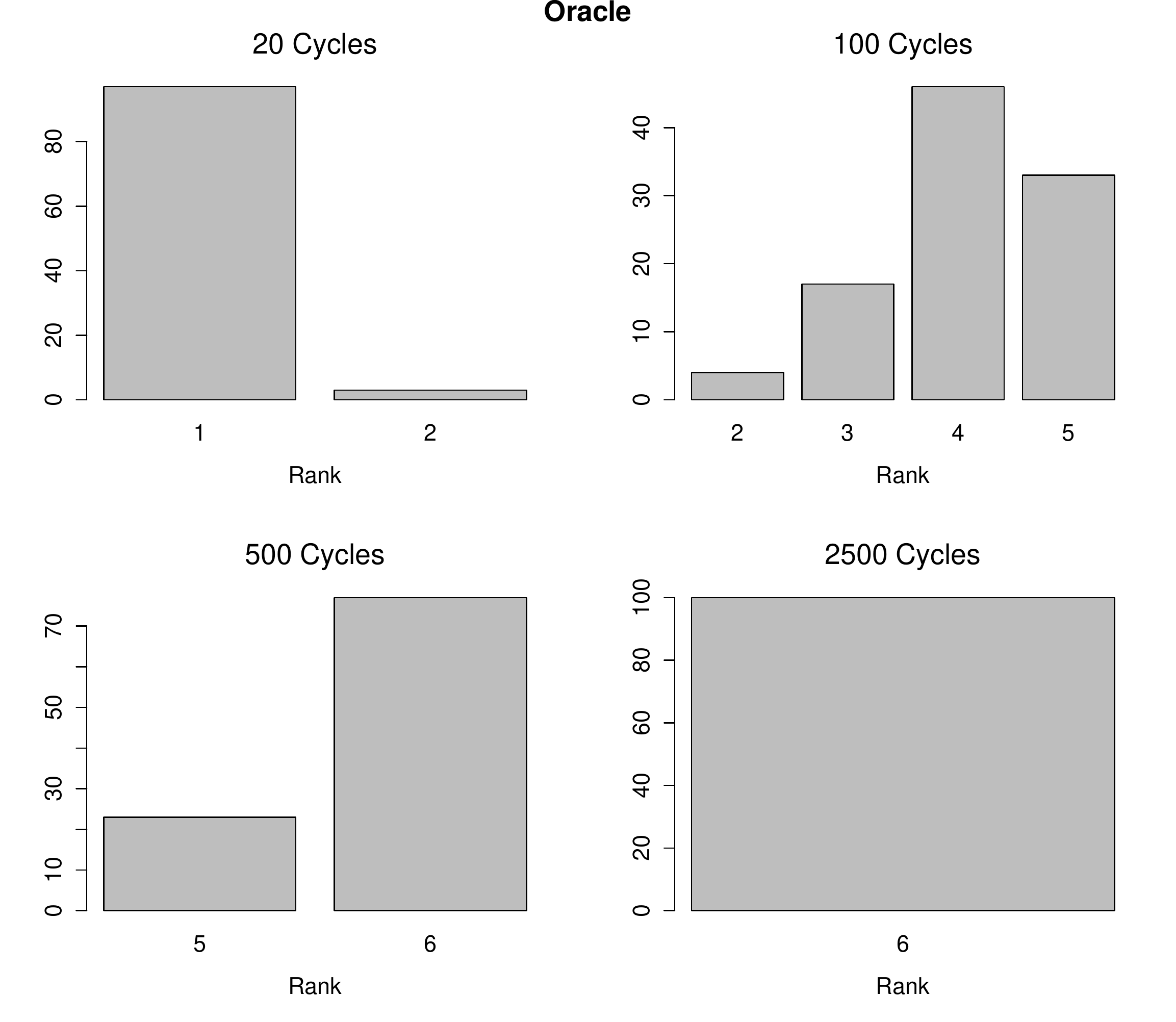}\\
a) Oracle estimator
 \end{minipage}
 \quad
 \begin{minipage}[b]{0.45\linewidth}
 \centering
\includegraphics[width=.9\linewidth]{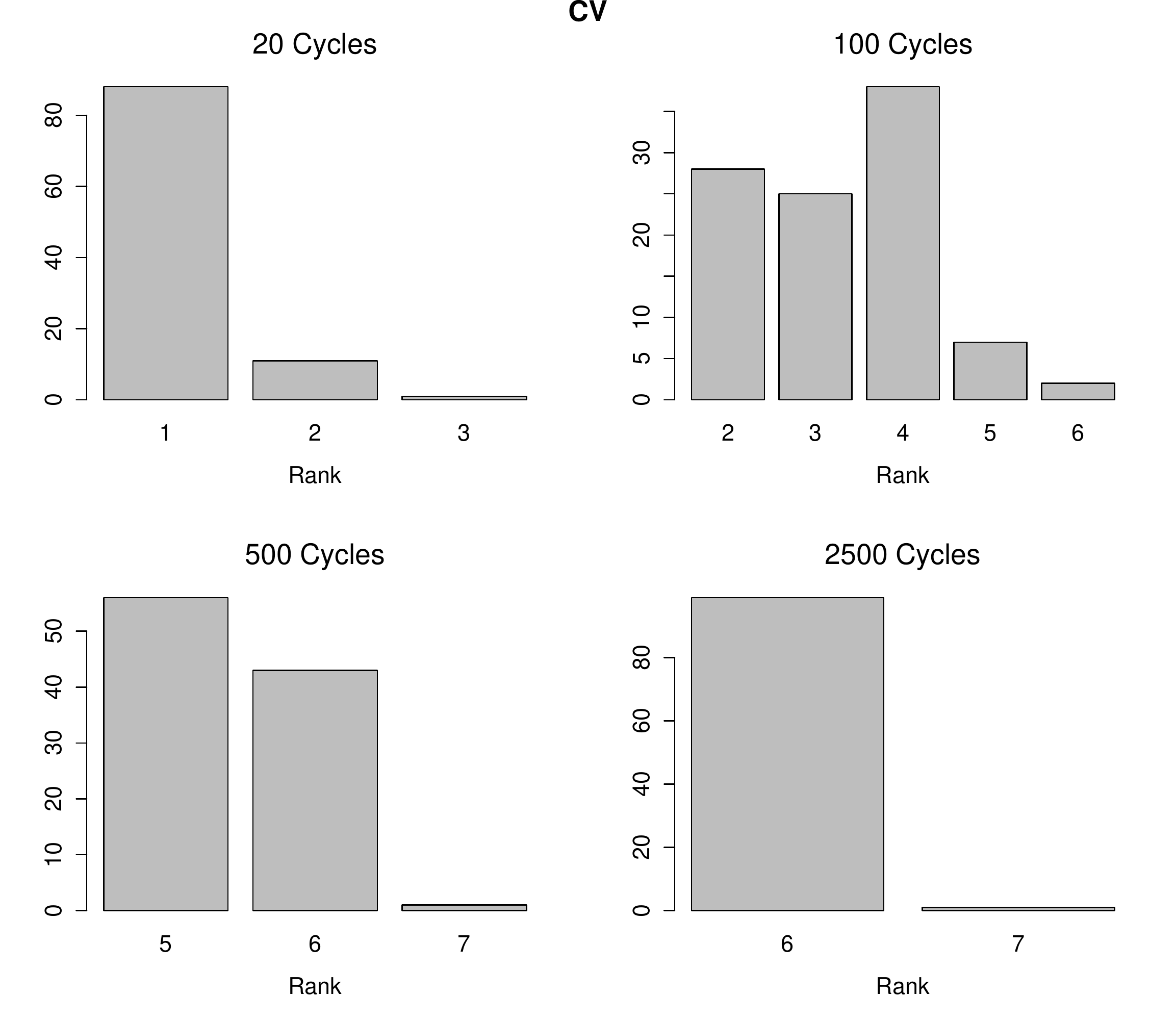}\\
b) Cross-validation estimator   
\end{minipage}

\vspace{4mm}

\begin{minipage}[b]{0.45\linewidth}
\centering
\includegraphics[width=.9\linewidth]{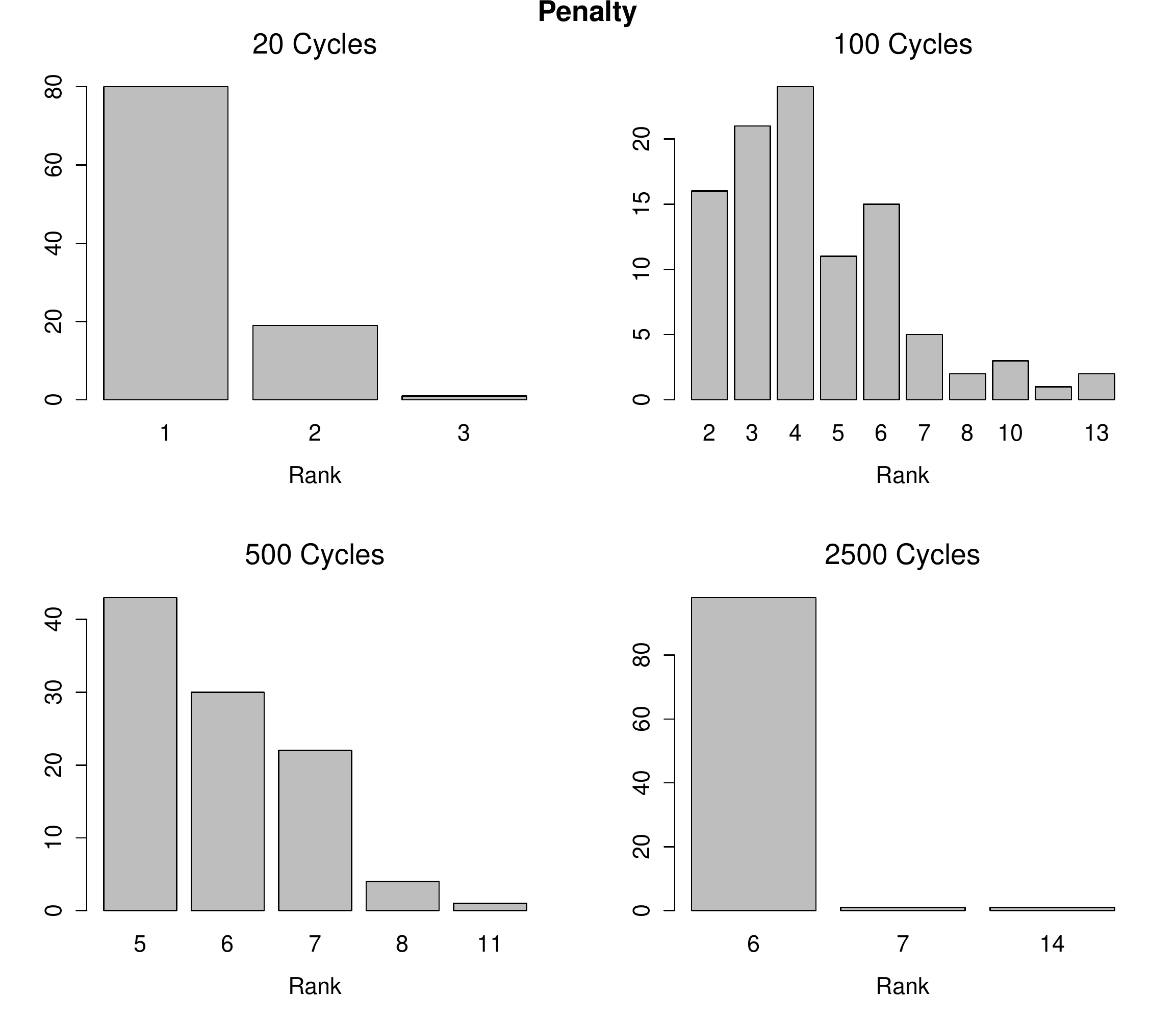}\\
c) Penalty estimator   
 \end{minipage}
 \quad
 \begin{minipage}[b]{0.45\linewidth}
 \centering
\includegraphics[width=.9\linewidth]{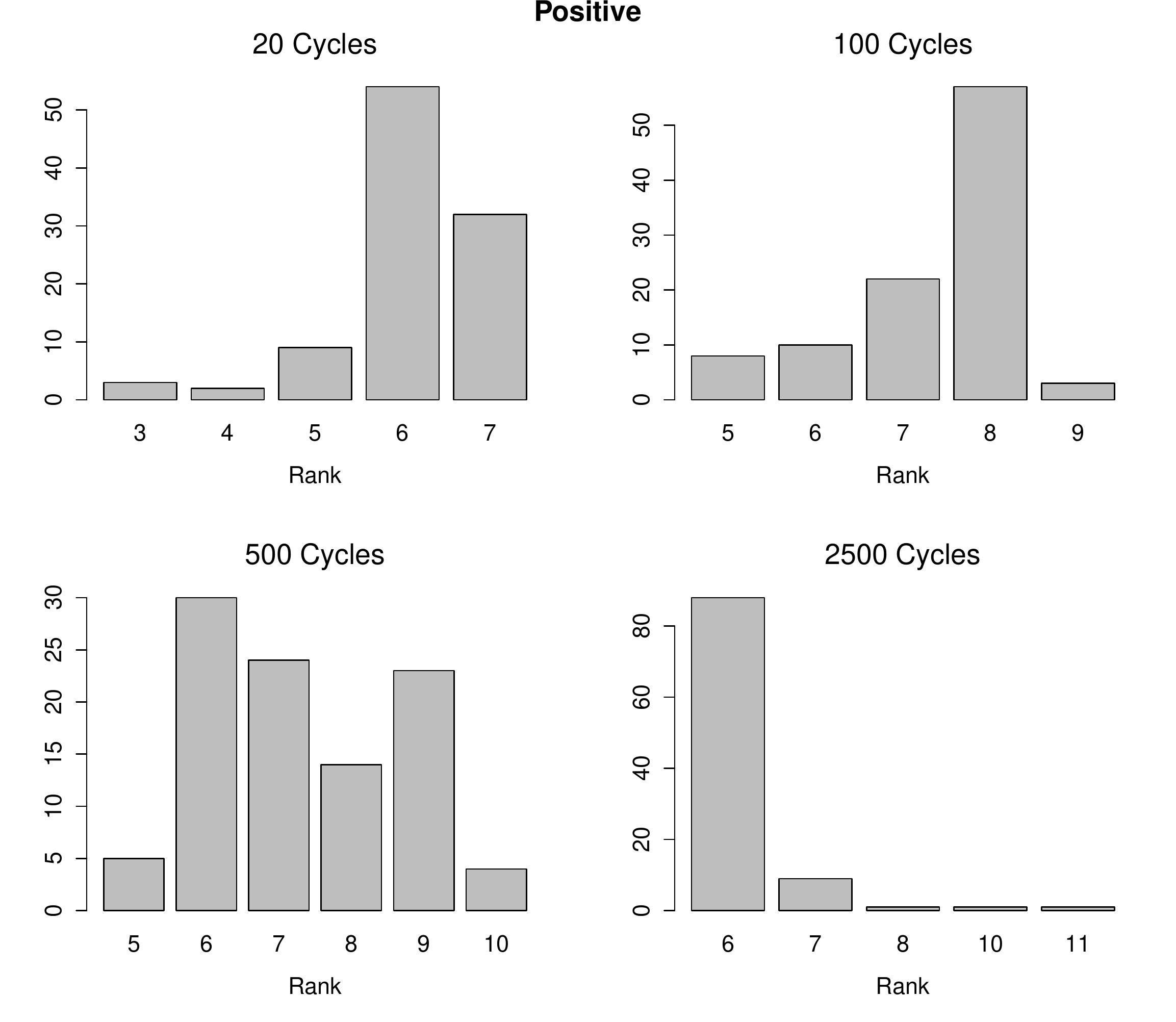}\\
d) Positive estimator   
\end{minipage}
\caption{Histograms of the empirical frequencies of the chosen rank for different estimators, as function of the number of repetitions $n$, true rank $r=6$, $k=4$ ($d=16$)}
\label{fig:histogram.chosen.rank}
\end{center}
\end{figure}

\vfill

\section{Conclusions and outlook}

Quantum state tomography, and in particular multiple ions tomography is an important enabling component of quantum engineering experiments. Since full quantum tomography becomes unfeasible for large dimensional systems, it is useful to identify lower dimensional models with good approximation properties for physically relevant states, and to develop estimation methods tailored for such models. In particular, quantum states created in the lab are often very well approximated by low rank density matrices, which are characterised by a number of parameter which is linear rather than quadratic in the space dimension. 

In this work we analysed several estimation algorithms targeted at estimating low rank states in multiple ions tomography. The procedure consists in computing the least squares estimator, which is then diagonalised, truncated to an appropriate smaller rank by setting eigenvalues below a ``noise threshold" to zero, and normalised. Among the several truncation methods proposed, the best performing one is the ``physical estimator"; this chooses the density matrix whose non-zero eigenvalues are above a certain threshold and is the closest to the least squares estimator. We proved concentration bounds and upper bounds for the mean square error of the penalised and physical estimators, as well as a lower bound for the asymptotic minimax rate for multiple ions tomography. The results show that the proposed methods have an 
optimal dependence on rank and dimension, up to a logarithmic factor in dimension. In addition, the algorithms are easy to implement numerically and their computational complexity is determined by that of the least squares estimator. 

An interesting future direction is to extend the spectral thresholding methodology to a measurement setup where a smaller number of settings is measured, which is however sufficient to identify the unknown low rank state. Another direction involves the construction of confidence intervals / regions for such estimators, beyond the  concentration bounds established here.

\emph{Acknowledgements.} M.G.'s work was supported by the EPSRC Grant No. EP/J009776/1.

\section{Appendix}\label{sec.appendix}

\subsection{Lemma on ${\bf A}^* \cdot {\bf A}$}

\begin{lem}\label{lemmaA}
Let ${\bf A}$ be the linear map defined in equation \eqref{eq.linear.map.a}. 
Then ${\bf A}^* \cdot {\bf A}$ is diagonal and
$$
[{\bf A}^* \cdot {\bf A}]_{{\bf b},{\bf b}} = 2^k 3^{d({\bf b})}, \mbox{ for all }
{\bf b} \in \{I,x,y,z\}^k.
$$
\end{lem}

\begin{proof}
Next, we give here for the reader's convenience the proof of Lemma~\ref{lemmaA} similar to that in \cite{AlquierBut2013}. Let us compute
\begin{equation*}\label{aux2}
[{\bf A}^* \cdot {\bf A}]_{{\bf b},{\bf b}'} = \sum_{\bf s} \sum_{\bf o} A_{\bf b}({\bf o}|{\bf s}) A_{{\bf b}'}({\bf o}|{\bf s}) .
\end{equation*}
If ${\bf b} = {\bf b}'$ it is easy to see that 
$$
[{\bf A}^* \cdot {\bf A}]_{{\bf b},{\bf b}} = \sum_{\bf s} \sum_{\bf o} \prod_{j\not \in E_{\bf b}} I(s_j = b_j) = 2^k 3^{d( {\bf b} )}. 
$$
If ${\bf b} \ne {\bf b}'$, we have either $E_{\bf b} = E_{{\bf b}'}$ or $E_{\bf b} \not = E_{{\bf b}'}$. On the one hand, in case the sets $E_{\bf b}$ and $E_{{\bf b}'}$ are equal, we have
$$
A_{\bf b}({\bf o}|{\bf s}) A_{{\bf b}'}({\bf o}|{\bf s}) = \prod_{j\not \in E_{\bf b}} I(s_j = b_j)\cdot I(s_j = b'_j) .
$$
For each ${\bf s}$, the previous product is 0. Indeed, if different from 0 then $b_j=b'_j$ for all $j$ not in $E_{{\bf b}}$. As $b_j=b_{j'}=I$ for $j$ in $E_{\bf b}$, it implies that ${\bf b} = {\bf b}'$ which contradicts the assumption here.

On the other hand, if the sets $E_{\bf b}$ and $E_{{\bf b}'}$ are different, there exists at least one coordinate $j_0$ in the symmetric difference $E_{\bf b} \Delta E_{b'}$ and the sum over outcomes ${\mathbf o}$ will split over values of ${\bf o}$ where $o_{j_0}=1$ and values where $o_{j_0}=-1$:
\begin{eqnarray*}
&& \sum_{\bf o} A_{\bf b}({\bf o}|{\bf s}) A_{{\bf b}'}({\bf o}|{\bf s})\\
&=&  \sum_{{\bf o}:o_{j_0}=1} I(s_{j_0} = b_{j_0}) \prod_{j\not \in E_{\bf b}}o_jI(s_j = b_j)\cdot \prod_{l\not \in E_{{\bf b}'}}o_l I(s_l = b'_l) \\
&&- \sum_{{\bf o}:o_{j_0}=-1} I(s_{j_0} = b_{j_0})\prod_{j\not \in E_{\bf b}}o_jI(s_j = b_j)\cdot \prod_{l\not \in E_{{\bf b}'}}o_l I(s_l = b'_l) =0.
\end{eqnarray*}
We assumed here that $j_0$ belongs to $E_{\bf b} \backslash E_{{\bf b}'}$ and the same holds for $j_0$ in $E_{{\bf b}'} \backslash E_{\bf b}$.
\end{proof}
\subsection{Proof of Proposition \ref{prop:linear}}
\label{app.proof.prop.linear}
\begin{proof}[Proof of Proposition \ref{prop:linear}]
Note that $f({\bf o}|{\bf s}) = \frac 1n \sum_i I(X_{{\bf s},i} = {\bf o})$, where the random variables $X_{{\bf s},i}$ are independent for all settings ${\bf s}$ and all $i$ from 1 to $n$. To estimate the risk of the linear estimator we write
\begin{eqnarray*}
\widehat \rho^{(ls)}_n - \rho 
& = & \sum_{\bf b} \sum_{\bf o} \sum_{\bf s} (f({\bf o}|{\bf s})) - p({\bf o}|{\bf s})) \frac{A_{\bf b}({\bf o}|{\bf s}) }{2^k 3^{d({\bf b})}}\sigma_{\bf b}\\
& = & \sum_{\bf b} \sum_{\bf o} \sum_{\bf s} \frac 1n \sum_i (I(X_{{\bf s},i}={\bf o}) - p({\bf o}|{\bf s}))\frac{A_{\bf b}({\bf o}|{\bf s}) }{2^k 3^{d({\bf b})}} \sigma_{\bf b}\\
& := & \sum_{\bf s} \sum_i W_{{\bf s},i}.
\end{eqnarray*}
where $W_{{\bf s},i}$ are independent and centered Hermitian random matrices. We will apply the following extension of the Bernstein matrix inequality \cite{AhlswedeWinter2002} due to \cite{tropp}, see also \cite{Koltchinsky}, \cite{Gross2011}.

\begin{prop}[Bernstein inequality, \cite{tropp}]\label{Bernstein}
Let $Y_1,...,Y_n$ be independent, centered, $m\times m$ Hermitian random matrices. Suppose that, for some constants $V,\, W>0$ we have $\|Y_j\|\leq V$, for all $j$ from 1 to $n$, and that $\|\sum_j \mathbb{E}(Y_j^2)\|\leq W$. Then, for all $t\geq 0$,
$$
\mathbb{P}(\|Y_1+...+Y_n\|\geq t) \leq 2 m \exp \left(- \frac{t^2/2}{W + tV/3} \right).
$$
\end{prop}
 
In our setup  we bound $\|W_{{\bf s},i}\| \leq V$ for all ${\bf s}$ and $i$ and $\|\sum_{\bf s} \mathbb{E} (W_{{\bf s},i}^* W_{{\bf s},i})\|\leq W$, where $V,\, W$ are evaluated below. We have 
\begin{eqnarray*}
\|W_{{\bf s},i}\| &\leq & \frac 1n\sum_{\bf b} \sum_{\bf o} \left|\frac{A_{\bf b}({\bf o}|{\bf s}) }{2^k 3^{d({\bf b})}}\right|\cdot |I(X_{{\bf s},i}={\bf o}) -p({\bf o}|{\bf s}) |\cdot \|\sigma_{\bf b}\|\\
& \leq & \frac 1n \sum_{\bf b} \frac 1{2^k 3^{d(b)}} \prod_{j \not \in E_{\bf b}} I(b_j = s_j)\sum_{\bf o} |I(X_{{\bf s},i}={\bf o})  - p({\bf o}|{\bf s})|\\
& \leq & \frac 2{n 2^k } \sum_{\ell = 0}^k \sum_{b: d(b) = \ell} \frac 1{3^\ell}
= \frac 2{n 2^k} \sum_{\ell=0}^k {k \choose \ell} \frac 1{3^\ell}
 = \frac 2{n2^k} \left(1+\frac 13\right)^k = \frac 2n \left(\frac 23\right)^k := V.
\end{eqnarray*}

Let us denote $B({\bf o}|{\bf s}) := \sum_{\bf b} {2^{-k} 3^{-d({\bf b})}} {A_{\bf b}({\bf o}|{\bf s}) }\sigma_{\bf b}$. Then  
\begin{eqnarray}\label{aux1}
&& \|\sum_{\bf s} \sum_i \mathbb{E} (W_{{\bf s},i}^* W_{{\bf s},i})\| \nonumber \\
& = & \frac 1{n^2}\| \sum_{\bf s}\sum_i \sum_{{\bf o},{\bf o'}} B^*({\bf o}|{\bf s})\cdot {\rm Cov}(I(X_{{\bf s},i}={\bf o}),I(X_{{\bf s},i}={\bf o'})) \cdot B({\bf o'}|{\bf s})\|\nonumber \\
&\leq & \frac 1{n^2} \|\sum_{\bf s} \sum_i \sum_{{\bf o}} B^*({\bf o}|{\bf s}) \cdot B({\bf o}|{\bf s})\|.
\end{eqnarray}
The last inequality follows from the fact that the covariance matrix can be written as the difference of two positive matrices
$
{\rm Cov}_{{\bf o},{\bf o'}}:= {\rm Cov}(I(X_{{\bf s},i}={\bf o}),I(X_{{\bf s},i}={\bf o'})) = p({\bf o}|{\bf s}) \delta_{{\bf o},{\bf o'}} - p({\bf o}|{\bf s})p({\bf o'}|{\bf s})
$ 
and since $p({\bf o}|{\bf s}) \leq 1$, we get ${\rm Cov} \leq \mathbf{1}$.

We replace $B({\bf o}|{\bf s})$ in (\ref{aux1}) and use Lemma~\ref{lemmaA} to get 
\begin{eqnarray*}
&& \|\sum_{\bf s} \sum_i \mathbb{E} (W_{{\bf s},i}^* W_{{\bf s},i})\| \\
&\leq & \frac 1n  \left\|\sum_{\bf s} \sum_{{\bf o}} \sum_{{\bf b}, {\bf b}'} \frac 1{2^{2k} 3^{d({\bf b}) + d({\bf b}')}} A_{\bf b}({\bf o}|{\bf s}) A_{{\bf b}'}({\bf o}|{\bf s}) \cdot \sigma_{\bf b} \sigma_{\bf b'} \right\|\\
&\leq & \frac 1{n 2^k}  \sum_{\ell = 0}^k \sum_{{\bf b}: d({\bf b}) = \ell} \frac 1{3^\ell}
= \frac 1{n 2^k} \sum_{\ell = 0}^k {k \choose \ell} \frac 1{3^\ell} = \frac 1n \left(\frac 23 \right)^k:= W.
\end{eqnarray*}

Apply the matrix Bernstein inequality in Proposition~\ref{Bernstein} to get, for any $t>0$:
$$
\mathbb{P}_\rho\left(\|\widehat \rho^{(ls)}_n - \rho \| \geq t \right) \leq 
2^{k+1} \exp \left(- \frac{t^2 }{1 + 2t/3} \frac n 2 \left(\frac 32\right)^k \right).
$$
We choose $t>0$ such that 
$$
\varepsilon = 2^{k+1} \exp\left( -\frac{t^2 }{1 + 2t/3} \cdot \frac n2 \left(\frac 32\right)^k\right),
$$
which leads to $t$ such that
$$
\frac{t^2 }{1 + 2t/3} = \frac{2}{n} \left( \frac 23\right)^k \log\left(\frac{2^{k+1}}{\varepsilon}\right) = \frac 32 v(\varepsilon).
$$
Then the convenient choice of $t$ is $ \nu(\varepsilon)= v(\varepsilon)/2 + \sqrt{v^2(\varepsilon)/2+3/2 \cdot v(\varepsilon)} $, which is equivalent to $\sqrt{3/2 \cdot v(\varepsilon)}$ when this last term tends to 0.
\end{proof}

\subsection{Proof of Proposition \ref{prop:trace}}
\label{app.prop:trace}

\begin{proof}[Proof of Proposition~\ref{prop:trace}]
Let us denote by $(\widehat \lambda_1 \geq ... \geq \widehat \lambda_d)$ the eigenvalues of 
$|\widehat \rho_n^{(ls)}|$ and by $\widetilde \lambda_1,...,\widetilde \lambda_d$ the eigenvalues of the resulting estimator $\widehat \rho_n^{(ls,n)}$. It is easy to see that the latter has the same eigenvectors as $\widehat \rho_n^{(ls)}$. 
Therefore
$$
(\widetilde \lambda_1,...,\widetilde \lambda_d) = \underset{\lambda :\sum_{j=1}^d \lambda_j = 1}
{\arg\,\,\min}\,\,\sum_{j=1}^d (\lambda_j - \widehat \lambda_j)^2.
$$
This optimisation has the explicit solution 
$$
\widetilde \lambda_j = \widehat \lambda_j +\frac L2, \, \mbox{where } \frac L2 d = 1- \sum_{j=1}^d \widehat \lambda_j.
$$
Note that 
$$
\frac L2 = \frac 1d \rm{tr}(\rho - \widehat \rho_n^{(ls)}) \leq \|\widehat \rho_n^{(ls)} - \rho\|.
$$
Therefore, $\|\widehat \rho_n^{(ls,n)} - \widehat \rho_n^{(ls)}\| = L/2 \leq \|\widehat \rho_n^{(ls)} - \rho\|$ and thus
\begin{eqnarray*}
\|\widehat \rho_n^{(ls,n)} - \rho\| &\leq & \|\widehat \rho_n^{(ls)} - \rho\| +\|\widehat \rho_n^{(ls,n)} - \widehat \rho_n^{(ls)}\| \leq 2 \nu(\varepsilon).
\end{eqnarray*}

\end{proof}

\subsection{Proof of Corollary \ref{cor:rk,state}}
\label{sec:7.4}

\begin{proof}[Proof of Corollary~\ref{cor:rk,state}]
We order the eigenvalues of $\widehat \rho_n^{(pen)}$ in decreasing order of absolute values 
$(|\widehat \lambda_1| \geq ... \geq |\widehat \lambda_d| )$, and denote by $\widetilde \lambda_1,...,\widetilde \lambda_d$ the eigenvalues of $\widehat \rho_n^{(pen,n)}$. As in the previous proof, we can see that both matrices will share the same eigenvectors and that the relation between the eigenvalues is
$$
\widetilde \lambda_j = \widehat \lambda _j +\frac L2, \mbox{ where } \frac L2 d = 1 - \sum_{j=1}^d \widehat \lambda_j.
$$
Thus, 
$$
\|\widehat \rho_n^{(pen)} - \widehat \rho_n^{(pen,n)} \|_F^2 = d (L/2)^2 \leq {\rm Tr}^2(\widehat \rho_n^{(pen)} - \rho) /d \leq \|\widehat \rho_n^{(pen)} -\rho\|_F^2 .
$$
We deduce that $\| \rho - \widehat \rho_n^{(pen,n)} \|_F \leq \|\widehat \rho_n^{(pen)} - \widehat \rho_n^{(pen,n)} \|_F + \|\widehat \rho_n^{(pen)} - \rho \|_F \leq 2  \|\widehat \rho_n^{(pen)} - \rho \|_F .$
Therefore, 
$$
\mathbb{P}\left(\| \widehat \rho_n^{(pen,n)} - \rho \|_F  > 8 c(\theta) r \nu(\varepsilon)^2\right) < \varepsilon
$$
and the previous inequality remains true for all $0 < e \leq \varepsilon < 1$.
Let us denote by 
$$
x(e) = 8 r \nu^2(e) = 16 \frac{rd}N \log (\frac{2d}e), 
$$ which is a decreasing function of $e$. This implies that $e = 2d \exp(- \frac{N x(e)}{16 rd})$. Thus,
\begin{eqnarray*}
\mathbb{E}_\rho\|\widehat \rho_n^{(pen,n)} - \rho\|_2^2 &=& \int_0^\infty \mathbb{P}_\rho(\|\widehat \rho_n^{(pen,n)} - \rho\|_2^2 > x) dx \\
& = & \int_0^{x(\varepsilon)} \mathbb{P}_\rho(\|\widehat \rho_n^{(pen,n)} - \rho\|_2^2 >x) dx + \int_{x(\varepsilon)}^\infty \mathbb{P}_\rho(\|\widehat \rho_n^{(pen,n)} - \rho\|_2^2 >x) dx\\
&\leq & x(\varepsilon ) + \int_{x(\varepsilon )}^\infty 2 d \exp(- \frac{N x }{16 rd}) dx\\
&\leq & 16 \frac{rd}N \log (\frac{2d}\varepsilon ) + 2 d \cdot 16\frac{rd}{N} \exp(-\frac{N  }{16 rd} x(\varepsilon))\\
& \leq & 16 \frac{rd}N \left( \log (\frac{2d}\varepsilon ) + \varepsilon\right)
\leq C \frac{rd}N \log (\frac{2d}\varepsilon ).
\end{eqnarray*}
\end{proof}
\subsection{Proof upper bound physical estimator}\label{proof.threshold}
\begin{proof}[Proof of Theorem~\ref{thm:recstate}]
We recall from Proposition~\ref{prop:trace} that with probability larger than $1-\varepsilon$, we have $\| \widetilde \rho_n^{(ls)} - \rho \| \leq 2 \nu(\varepsilon)$. In particular, by using the Weyl inequality \cite{Horn}, this implies that $ |\widetilde \lambda_k - \lambda_k| \leq 2 \nu(\varepsilon)$ for all $k$ from 1 to $d$, where $\lambda_1,...,\lambda_d$ are  the eigenvalues of $\rho $ arranged in decreasing order.

After a total of  $\hat{\ell} = d-\hat{r}$ iterations the algorithm stops and we have \cite{Smolin} 
$$
\widehat \lambda_{\hat{r}+1}^{(phys)} =...= \widehat \lambda_d^{(phys)}=0
$$ 
and
$$
\widehat \lambda_j^{(phys)} :=\widehat \lambda_j^{(\hat\ell)} = \widetilde \lambda_j +
\frac{1}{\hat{r}} \sum_{ k > \hat{r}} \widetilde \lambda_k
\qquad j=1,\dots, \hat{r}.
$$
From now on we assume that $\|\widetilde \rho_n^{(ls)} - \rho\| \leq 2 \nu( \varepsilon)$ which holds with probability larger than $1-\varepsilon$, cf. Corollary \ref{rho.tilde.ls}. We will show that under this assumption $\hat{r} = r$. For this we consider the two cases  $ \hat{r} >r $ and $ \hat{r} <r$ separately.

Suppose $ \hat{r} >r $. For $j \leq \hat{r}$ we have
\begin{eqnarray*}
|\widehat \lambda_j^{(phys)} - \lambda_j|& = & \left|\widetilde \lambda_j - \lambda_j + \frac{1}{\hat{r}}
\left( 1- \sum_{ k \leq \hat{r}} \widetilde \lambda_k\right)\right|\\
& = & \left|\widetilde \lambda_j - \lambda_j + \frac{1}{\hat{r}} \sum_{ k \leq \hat{r}} (\lambda_k - \widetilde \lambda_k ) \right|\\
& \leq & |\widetilde \lambda_j - \lambda_j|+\frac {1}{\hat{r}} \sum_{ k \leq \hat{r}} |\widetilde \lambda_k-\lambda_k| \leq 2 \nu+  
2\nu = 4 \nu.
\end{eqnarray*}
This implies that $\widehat{\lambda}_{\hat{r}}^{(phys)} \leq \lambda_{\hat{r}} + 4\nu = 4\nu$ since $\hat{r} > r$. However, as discussed above, by the definition of the threshold estimator, $\widehat{\lambda}_{\hat{r}}^{(phys)} \geq 4\nu$. Therefore $\hat{r}$ cannot be larger than $r$.

Suppose $ \hat{r} <r $.  Then $\widehat \lambda_r^{(phys)} \geq \lambda_r - |\widehat \lambda_r^{(phys)}- \lambda_r| >8\nu-  4\nu = 4\nu$. However this is not possible since $\widehat \lambda_r^{(phys)}=0$ under the assumption 
$\hat{r} <r $. Thus, $\widehat r = r $ with probability larger than $1-\epsilon$.

We consider now the MSE of the estimator. As shown above, we have 
$|\widehat{\lambda}^{(phys)}_j - \lambda_j| \leq 4\nu$. Let $\rho= UDU^*$ be the diagonalisation of $\rho$, where $U$ is unitary and $D$ is the diagonal matrix of eigenvalues. Similarly, let 
$\widehat{\rho}^{(ls)}_n= \widehat{U}_n {D}^{(ls)}_n\widehat{U}^*_n$ and 
$\widehat{\rho}^{(phys)}_n= \widehat{U}_n {D}^{(phys)}_n\widehat{U}^*_n$ be the decompositions of the least squares and the physical threshold estimator, where we take into account that the latter two have the same eigenbasis.
Then
\begin{eqnarray}
\|\widehat \rho_n^{(phys)} - \rho\|_2 &=& 
\|\widehat{U}_n D_n^{(phys)} \widehat{U}_n^* - U D U^* \|_2 \nonumber\\
 &\leq& 
\|\widehat{U}_n D_n^{(phys)} \widehat{U}_n^*  - \widehat{U}_n D \widehat{U}_n^*\|_2+
\|\widehat{U}_n D \widehat{U}_n^* - U D U^* \|_2 \nonumber\\
&\leq&
\|D_n^{(phys)} - D\|_2 +\|\widehat{U}_n D \widehat{U}_n^* - U D U^* \|_2\label{eq.triangle1}
\end{eqnarray}
The first norm on the right-hand side of the previous inequality is bounded as 
\begin{equation}\label{eq.diag.2}
\|D_n^{(phys)} - D\|_2^2 = \sum_{j \leq \widehat r} (\widehat \lambda_j^{(phys)}-\lambda_j)^2 \leq 16 {r} \nu^2.
\end{equation}
For the second norm we use a similar triangle inequality for the \emph{operator norm} 
\begin{eqnarray*}
\|\widehat{U}_n D \widehat{U}_n^*-  U D  U^*\| 
&\leq & \|\widehat{U}_n D \widehat{U}_n^* - \widehat{U}_n  D_n^{(ls)} \widehat{U}_n^*\| 
+ \|\widehat{U}_n  D_n^{(ls)} \widehat{U}_n^* - U D U^*\| \\
&=& \| D  -   D_n^{(ls)}\|  + \|\widehat\rho_n^{(ls)} - \rho\| \leq \nu + \nu = 2\nu.
\end{eqnarray*}
The first term is smaller than $\nu$ since $|\lambda_i - \widehat{\lambda}^{(ls)}_i|\leq \nu$ as it follows from Proposition \ref{prop:linear}, and the Weyl inequality \cite{Horn}. The second term is also bounded by $\nu$, by Proposition \ref{prop:linear}. Therefore, since $\widehat{U}_n D \widehat{U}_n^*$ and $U D  U^*$ are rank $r$ matrices, the difference is at most rank 
$2r$ and $\| \widehat{U}_n D \widehat{U}_n^* - U D  U^* \|_2 \leq 2\nu \sqrt{2 r}$. By plugging this together with \eqref{eq.diag.2} into the right side of \eqref{eq.triangle1} we get
$$
\|\widehat \rho_n^{(phys)} - \rho\|^2_2 \leq (2\sqrt{2}+ 4)^2 r\nu^2,
$$
with probability larger than $1-\epsilon$.

Moreover, the previous inequality remains true for all $0 < e \leq \varepsilon < 1$.
Let us denote by 
$$
x(e) = 48 r \nu^2(e) = 96 \frac{rd}N \log (\frac{2d}e), 
$$ which is a decreasing function of $e$. This implies that $e = 2d \exp(- \frac{N x(e)}{96 rd})$. Thus,
\begin{eqnarray*}
\mathbb{E}_\rho\|\widehat \rho_n^{(phys)} - \rho\|_2^2 &=& \int_0^\infty \mathbb{P}_\rho(\|\widehat \rho_n^{(phys)} - \rho\|_2^2 > x) dx \\
& = & \int_0^{x(\varepsilon)} \mathbb{P}_\rho(\|\widehat \rho_n^{(phys)} - \rho\|_2^2 >x) dx + \int_{x(\varepsilon)}^\infty 
\mathbb{P}_\rho(\|\widehat \rho_n^{(phys)} - \rho\|_2^2 >x) dx\\
&\leq & x(\varepsilon ) + \int_{x(\varepsilon )}^\infty 2 d \exp(- \frac{N x }{96 rd}) dx\\
&\leq & 96 \frac{rd}N \log (\frac{2d}\varepsilon ) + 2 d \cdot 96\frac{rd}{N} \exp(-\frac{N  }{96 rd} x(\varepsilon))\\
& \leq & 96 \frac{rd}N \left( \log (\frac{2d}\varepsilon ) + \varepsilon)\right)
\leq C \frac{rd}N \log (\frac{2d}\varepsilon ).
\end{eqnarray*}
\end{proof}


\subsection{The average Fisher information matrix for the full, unconstrained model, with random measurement design} 
\label{sec.averageFisher}

In this section we present a detailed calculation of the average quantum Fisher information at a the rank  
$r$ state $\rho_0$ defined in \eqref{eq.rho0}, for random measurements with uniform distribution over the measurement basis. We consider the full parametrisation by $\theta\in \mathbb{R}^{d^2}$ given by equation \eqref{eq.def.theta}. As explained in the proof, the Fisher information matrix for the parametrisation  $\tilde{\theta}$ of the rotation model $\mathcal{R}_{d,r}$ is a particular block of 
the larger Fisher matrix computed here. We will come back to this at the end of the computation.

Let 
$$
{\bf B}_U := \{|{\bf o};U\rangle := U|{\bf o}\rangle   : {\bf o} =1,\dots, 2^k \}
$$
denote the ONB obtained by rotating the standard basis by the unitary $U$. 
The Fisher info $I(\rho| {\bf B}_U)$ associated to the measurement with ONB 
${\bf B}_U$ is given by the matrix 
$$
I(\rho| {\bf B}_U)_{a,b} = \sum_{{\bf o}: p_\rho({\bf o}|{\bf B}_U)>0} \frac{1}{p_\rho({\bf o}|{\bf B}_U)}
\frac{\partial p_\rho({\bf o}|{\bf B}_U)}{\partial \theta_a}
\frac{\partial p_\rho({\bf o}| {\bf B}_U)}{\partial \theta_b}
$$
where $p_\rho({\bf o}|{\bf B}_U)$ is the probability of the outcome ${\bf o}$
\begin{eqnarray*}
p_\rho({\bf o}|{\bf B}_U) &=& \langle {\bf o},U |\rho |{\bf o}, U\rangle = \sum_{i} \rho_{ii} |\langle  {\bf o},U |i\rangle|^2 \\
&+& 
2 \sum_{i<j}{\rm Re} (\rho_{i,j}) {\rm Re} (   \langle i|  {\bf o},U\rangle\langle   {\bf o},U|j\rangle)
+ 2 \sum_{i<j} {\rm Im} (\rho_{i,j}) {\rm Im} (  \langle i|  {\bf o},U\rangle\langle   {\bf o},U|j\rangle)\\
\end{eqnarray*}
and the partial derivatives are given by
\begin{eqnarray*}
\frac{\partial p_\rho({\bf o}|{\bf B}_U) }{\partial \rho_{i,i}} &=& |\langle  {\bf o},U |i\rangle|^2\\
\frac{\partial p_\rho({\bf o}|{\bf B}_U) }{\partial {\rm Re} \rho_{i,j}} &=& 2{\rm Re} (   \langle i|  {\bf o},U\rangle\langle   {\bf o},U|j\rangle)\\
\frac{\partial p_\rho({\bf o}|{\bf B}_U) }{\partial {\rm Im} \rho_{i,j}} &=& 2{\rm Im} (   \langle i|  {\bf o},U\rangle\langle   {\bf o},U|j\rangle).
\end{eqnarray*}
The Fisher information matrix $I(\rho| {\bf B}_U)$ has the following block structure
$$
I(\rho| {\bf B}_U)=\left(
\begin{array}{ccccc}
I^{dd}(\rho| {\bf B}_U) && I^{dr}(\rho| {\bf B}_U) && I^{di}(\rho| {\bf B}_U) \\ 
&&&&\\
I^{rd}(\rho| {\bf B}_U) && I^{rr}(\rho| {\bf B}_U) && I^{ri}(\rho| {\bf B}_U) \\ 
&&&&\\
I^{id}(\rho| {\bf B}_U) && I^{ir}(\rho| {\bf B}_U) && I^{ii} (\rho| {\bf B}_U)
\end{array}\right)
$$
with superscripts indicating the 
type of parameter considered: diagonal, real or imaginary part of off-diagonal element.
The average Fisher information for a randomly chosen basis is
$$
\bar{I} (\rho):= \int  \mu(dU ) I(\rho|{\bf B}_U)
$$
where $\mu(dU)$ is the Haar measure over unitaries used for choosing the random basis. Note that by symmetry 
$\bar{I} (\rho)$ only depends on the spectrum of $\rho$. We will not compute $\bar{I} (\rho)$ for an arbitrary state $\rho$ but only at 
$\rho=\rho_0$. The corresponding Fisher information will be denoted $\bar{I} = \bar{I} (\rho_0)$ and is a function of $d$ and $r$. 
Below we compute the different blocks of $\bar{I}$. For the matrix elements we will use a suggestive notation, e.g. 
$\bar{I}^{d,r}_{ii;jk}$ denotes the element corresponding to the diagonal parameter $\theta^d_{ii} = \rho_{ii}$ and the real part of the off-diagonal element $\rho_{jk}$, etc. 

\noindent
\emph{A) Diagonal-diagonal block.}

$$
I(\rho|{\bf B}_U)^{dd}_{ii;jj} =  \sum_{{\bf o}: p_\rho({\bf o}|{\bf B}_U)>0} 
\frac{1}{p_\rho({\bf o}|{\bf B}_U)}  |\langle  {\bf o},U |i\rangle|^2  |\langle  {\bf o},U |j\rangle|^2
$$
By integrating over unitaries we obtain the corresponding matrix element of the average Fisher matrix $\bar{I}$. Since 
$p_{\rho_0}({\bf o}|{\bf B}_U)>0$ is true for all ${\bf o}$, with probability one with respect to $\mu^d(dU)$, 
we drop the condition from the sum. At the state $\rho= \rho_0$ defined in \eqref{eq.rho0}, we have
\begin{eqnarray*}
\bar{I}^{dd}_{ii;jj} &=& r \sum_{\bf o} 
\int  \frac{|\langle  {\bf o},U |i\rangle|^2  |\langle  {\bf o},U |j\rangle|^2 }{\sum_{k=1}^r | \langle {\bf o} , U| k\rangle|^2} \mu^d(dU)\\
&=&
r\cdot d 
\int  \frac{|\langle  {\bf 1},U |i\rangle|^2  |\langle  {\bf 1},U |j\rangle|^2 }{\sum_{k=1}^r | \langle {\bf 1} , U| k\rangle|^2} \mu^d(dU)\\
&=&
r\cdot d \int \frac{|\langle  \psi |i\rangle|^2  |\langle  \psi |j\rangle|^2 }{\sum_{k=1}^r | \langle \psi| k\rangle|^2} \nu^d(d\psi)
\end{eqnarray*}
where $\nu^d(d\psi)$ is the uniform measure over the projective space on $\mathbb{C}^d$. 
To compute the integral we decompose $|\psi\rangle $ as
$$
|\psi\rangle = q |\psi_1\rangle+ \sqrt{1-q^2}|\psi_2\rangle 
$$
with $|\psi_1\rangle \in\mathcal{H}_r:={\rm Span} \{|1\rangle, \dots , |r\rangle  \}$ and 
$|\psi_2\rangle\in  \mathcal{H}^\perp_r$ normalised (orthogonal) vectors, and $0\leq q\leq 1$. The uniform measure can be expressed as
$$
\nu^d(d\psi) =  m^{d,r}(dq) \times \nu^r(d\psi_1) \times   \nu^{d-r}(d\psi_2)
$$ 
where $m^{d,r}(dq) $ is the distribution of the length of the projection of a random vector in $\mathbb{C}^d$ 
onto an $r$-dimensional subspace. With this notation we have
\begin{equation}\label{eq.fisher.dd}
\bar{I}^{dd}_{ii;jj} = 
r\cdot d \int \int \int \frac{1}{q^2} |\langle  \psi |i\rangle|^2  |\langle  \psi |j\rangle|^2 m^{d,r}(dq)
 \nu^r(d\psi_1) \nu^{d-r}(d\psi_2).
\end{equation}
We distinguish $4$ sub-cases depending on whether each of the indices $i,j$ belongs to $\{1,\dots ,r \}$ or 
$\{ r+1, \dots, d\}$. 

\noindent
\emph{Sub-case 1: $i,j\leq r$.} In this case \eqref{eq.fisher.dd} becomes
\begin{eqnarray}
\bar{I}^{dd}_{ii;jj} &=& 
r\cdot d \int  \frac{1}{q^2} q^4 m^{d,r}(dq) \times \int |\langle  \psi_1 |i\rangle|^2  |\langle  \psi_1 |j\rangle|^2 
 \nu^r(d\psi_1) \nonumber \\
&=&
r\cdot d \int  \frac{1}{q^2} q^4 m^{d,r}(dq) \times \int  U^*_{1,i} U_{i,1} U^*_{1,j} U_{j,1} \mu^r (dU).
\label{eq.fisher.dd.case1}
\end{eqnarray}
In the last line we have re-written $|\psi_1\rangle$ as $U|1\rangle$ in order to use the existing formulas for integrals of monomials over the unitary group \cite{Collins}. Since we will use these formulas repeatedly, we recall that
\begin{equation}\label{eq.weingarten}
 \int  U_{i_1, j_1}  U_{i_2,j_2}  U^*_{j^\prime_1, i^\prime_1}U^*_{j^\prime_2, i^\prime_2}\mu^l (dU) = 
 \sum_{\sigma,\tau\in S_2}\delta_{i_1,i^\prime_{\sigma 1}}\delta_{i_2,i^\prime_{\sigma 2}} \delta_{j_1,j^\prime_{\tau 1}}\delta_{j_2,j^\prime_{\tau 2}} {\rm Wg}^l(\sigma\tau^{-1})
\end{equation}
where $\sigma$ and $\tau$ are permutations in $S_2 = \{ (1,1) , (2)\}$ and ${\rm Wg}^l(\cdot)$ is the Weingarten function over $S_2$:
$$
{\rm  {\rm Wg}}^l(1,1) =\frac{1}{l^2-1}, \qquad  
{\rm Wg}^l(2) =\frac{-1}{l(l^2 -1)}.
$$
The integral over $q$ can be easily evaluated as
\begin{equation}\label{eq.int.qsquare}
\int  q^2 m^{d,r}(dq) =  \sum_{k=1}^r \int |\langle k|\psi\rangle|^2 \mu^d(d\psi) = \frac{r}{d}. 
\end{equation}

By inserting \eqref{eq.weingarten} and \eqref{eq.int.qsquare} into \eqref{eq.fisher.dd.case1} we get 
\begin{eqnarray*}
\bar{I}^{dd}_{ii;jj} &=& 
r\cdot d 
\cdot \frac{r}{d}
\left(
\frac{1}{r^2-1}- \frac{1}{r(r^2 -1)}  
\right)=\frac{r}{r+1}, \quad i\neq j\\
\bar{I}^{dd}_{ii;ii} &=& 
r\cdot d 
\cdot \frac{r}{d}
\cdot 2\left(
\frac{1}{r^2-1}- \frac{1}{r(r^2 -1)}  
\right)=\frac{2r}{r+1}.\\
\end{eqnarray*}

\noindent
\emph{Sub-case 2: $i\leq r$ and $ j>r$.} 
From \eqref{eq.fisher.dd} we get
\begin{eqnarray*}
\bar{I}^{dd}_{ii;jj} &=& 
r\cdot d \int  \frac{1}{q^2} q^2(1-q^2) m^{d,r}(dq) \times \int |\langle  \psi_1 |i\rangle|^2   \nu^r(d\psi_1) \times \int |\langle  \psi_2 |j\rangle|^2 
 \nu^{d-r}(d\psi_2) \\
&=&
r\cdot d \left(1-\frac{r}{d}\right) \cdot \frac{1}{r} \cdot \frac{1}{d-r} = 1.
\end{eqnarray*}

\noindent
\emph{Sub-case 3: $i, j>r$.} Similarly, in this case
\begin{eqnarray}
\bar{I}^{dd}_{ii;jj} &=& 
r\cdot d \int  \frac{1}{q^2} (1-q^2)^2 m^{d,r}(dq) \times \int |\langle  \psi_2 |i\rangle|^2  |\langle  \psi_2 |j\rangle|^2 
 \nu^{d-r}(d\psi_2)\nonumber \\
&=&
r\cdot d \int  \frac{1}{q^2} (1-q^2)^2 m^{d,r}(dq) \times \int  U^*_{1,i} U_{i,1} U^*_{1,j} U_{j,1} \mu^{d-r} (dU)
\label{eq.dd.rr}
\end{eqnarray}
We first simplify the integral on the right side
\begin{equation}\label{eq.itegral.inverseq}
\int  \frac{1}{q^2} (1-q^2)^2 m^{d,r}(dq) = \int  \frac{1}{q^2} m^{d,r}(dq) -2+ \frac{r}{d}.
\end{equation}
To evaluate the remaining integral, we consider a multivariate Gaussian random variable 
$c= (c_1, \dots ,c_{2d})\sim N(0, I_{2d})$, and denote by $g^{2d}(dc)$ its probability distribution. From this we construct the complex vector with uniform distribution over the unit ball in $\mathbb{C}^d$
$$
|\psi\rangle := \frac{1}{\|c\|} \sum_{i=1}^d (c_i+ ic_{d+1} ) |i\rangle, \qquad \|c\|^2= \sum_{l=1}^{2d}c_i^2.
$$
We can now write
\begin{eqnarray}
 \int  \frac{1}{q^2} m^{d,r}(dq)  &=& \int g^{2d}(dc) \frac{\sum_{l=1}^{2d}{c^2_l}}{ \sum_{l=1}^{2r}{c^2_l}} = 1+
 \int g^{2d}(dc) \frac{\sum_{l=2r+1}^{2d}{c^2_l}}{ \sum_{l=1}^{2r}{c^2_l}} \nonumber \\
& =& 1 + \int \|c^1\|^2 g^{2(d-r)}(dc^1)  \cdot \int  \frac{1}{\| c^2 \|^2}g^{2r}(dc^2) \nonumber\\
&=&1+ 2(d-r) \cdot \frac{1}{2r-2} = \frac{d-1}{r-1}.\label{eq.inverse.chi}
\end{eqnarray}
Above we used the fact that $\|c^2\|^2$ is a $\chi^2$ variable with $2r$ degrees of freedom, so the second integral is the mean of its inverse.  By inserting \eqref{eq.inverse.chi} into \eqref{eq.itegral.inverseq} we obtain
\begin{equation}\label{eq.funny.integral}
\int  \frac{1}{q^2} (1-q^2)^2 m^{d,r}(dq)= -2 +\frac{r}{d} + \frac{d-1}{r-1} = \frac{(d-r)(d-r+1)}{d(r-1)}
\end{equation}

Finally, by inserting \eqref{eq.funny.integral} into \eqref{eq.dd.rr}  and applying \eqref{eq.weingarten} we obtain 
\begin{eqnarray*}
\bar{I}^{dd}_{ii;jj} &=& 
r\dot d 
\cdot \frac{(d-r)(d-r+1)}{d(r-1)}
\cdot 
\left(
\frac{1}{(d-r)^2-1}- \frac{1}{(d-r)((d-r)^2 -1)}  
\right)\\
&=& 
\frac{ r}{r-1},\quad i\neq j
\\
\bar{I}^{dd}_{ii;ii} &=& 
\frac{ 2r}{r-1}.
\end{eqnarray*}
Note that for $r=1$ these matrix elements are infinite. This is due to large contributions from measurements which have one basis vector close to being orthogonal to the one dimensional vector state.. This is somewhat akin to what happens in the case of a Bernoulli variable (coin toss), in the case when the probability is close to zero or one. While this phenomenon is interesting, it does not play any role in our analysis for which the diagonal matrix elements are not relevant parameters.

\vspace{2mm}

\noindent
\emph{B) Diagonal-Real block.}
$$
I(\rho|{\bf B}_U)^{dr}_{ii;jk} = 2 \sum_{{\bf o}: p_\rho({\bf o}|{\bf B}_U)>0} 
\frac{1}{p_\rho({\bf o}|{\bf B}_U)}  |\langle  {\bf o},U |i\rangle|^2  \cdot
{\rm Re} ( \langle  j| {\bf o},U \rangle \langle{\bf o},U| k\rangle)
$$
As before, the corresponding entry of the average Fisher information matrix is
\begin{eqnarray*}
\bar{I}^{dr}_{ii;jk} &=& 2\cdot d\cdot r \int \frac{|\langle  \psi |i\rangle|^2\cdot     
{\rm Re}( \langle  j| \psi \rangle \langle\psi | k\rangle )  }{\sum_{k=1}^r | \langle \psi| k\rangle|^2} \nu^d(d\psi)\\
&=& 
 2\cdot d\cdot r \int\!\!\!\!\int\!\!\!\!\int \frac{1}{q^2} |\langle  \psi |i\rangle|^2  \cdot     {\rm Re}( \langle  j| \psi \rangle \langle\psi | k\rangle )  m^{d,r}(dq)
 \nu^r(d\psi_1) \nu^{d-r}(d\psi_2).
\end{eqnarray*}

\noindent
\emph{Sub-case 1: $i,j,k \leq r$.}
In this case the integral is
\begin{eqnarray*}
\bar{I}^{dr}_{ii;jk} &=&  2\cdot d\cdot r \int q^2  m^{d,r}(dq)\int  |\langle  \psi_1 |i\rangle|^2   {\rm Re}( \langle  j| \psi_1 \rangle \langle\psi_1 | k\rangle )  \nu^r(d\psi_1) \\
&=&d\cdot r \cdot \frac{r}{d}  \int \mu^r(dU) U_{i1} U^*_{1i} ( U_{j1}U^*_{1k} +    U_{k1}U^*_{1j} ) =0,
\end{eqnarray*}
where we applied formula \eqref{eq.weingarten} for the integrals over the unitaries.

\noindent
\emph{Sub-case 2: $i,j \leq r $ and  $k>r$.}
In this case the integral is
\begin{eqnarray*}
\bar{I}^{dr}_{ii;jk} &=&   2\cdot d\cdot r \int q\sqrt{1-q^2}  m^{d,r}(dq) \\
&\times &
{\rm Re} \int  |\langle  \psi_1 |i\rangle|^2   \langle  j| \psi_1\rangle \nu^r(d\psi_1) \int  \langle\psi_2 | k\rangle   \nu^{d-r}(d\psi_2) =0.
\end{eqnarray*}
\emph{Sub-case 3: $i \leq r $ and $ j, k>r$.}
In this case the integral is
\begin{eqnarray*}
\bar{I}^{dr}_{ii;jk} &=& 2\cdot d\cdot r \int (1-q^2)  m^{d,r}(dq) \\
&\times &
 \int  |\langle  \psi_1 |i\rangle|^2  \nu^r(d\psi_1)\cdot {\rm Re} \int   \langle  j| \psi_2\rangle \langle\psi_2 | k\rangle   \nu^{d-r}(d\psi_2) =0.
\end{eqnarray*}
\emph{Sub-case 4: $i > r $ and $ j, k\leq r$.}
In this case the integral is 
\begin{eqnarray*}
\bar{I}^{dr}_{ii;jk} &=&  2\cdot d\cdot r \int (1-q^2)  m^{d,r}(dq) \\
&\times &
 \int  |\langle  \psi_2 |i\rangle|^2  \nu^{d-r}(d\psi_2)\cdot {\rm Re} \int   \langle  j| \psi_1\rangle \langle\psi_1 | k\rangle   \nu^{r}(d\psi_1) =0.
\end{eqnarray*}
\emph{Sub-case 5: $i > r $ and $ j, k\leq r$.}
In this case the integral is 
\begin{eqnarray*}
\bar{I}^{dr}_{ii;jk} &=&  2\cdot d\cdot r \int (1-q^2)  m^{d,r}(dq) \\
&\times &
 \int  |\langle  \psi_2 |i\rangle|^2  \nu^{d-r}(d\psi_2)\cdot {\rm Re} \int   \langle  j| \psi_1\rangle \langle\psi_1 | k\rangle   \nu^{r}(d\psi_1) =0.
\end{eqnarray*}
\emph{Sub-case 6: $i ,k> r $ and $j\leq r$.}
In this case the integral is 
\begin{eqnarray*}
\bar{I}^{dr}_{ii;jk} &=& 2\cdot d\cdot r \int q\sqrt{1-q^2}  m^{d,r}(dq) \\
&\times &
 {\rm Re} \int  |\langle  \psi_2 |i\rangle|^2 \langle\psi_2 | k\rangle  \nu^{d-r}(d\psi_2)\cdot  
 \int   \langle  j| \psi_1\rangle   \nu^{r}(d\psi_1) =0.
\end{eqnarray*}
\emph{Sub-case 7: $i ,j,k> r $.}
In this case the integral is 
\begin{eqnarray*}
\bar{I}^{dr}_{ii;jk} &=&  2\cdot d\cdot r \int \frac{(1-q^2)^2}{q^2}  m^{d,r}(dq) \\
&\times &
 \int  |\langle  \psi_2 |i\rangle|^2    {\rm Re}   \langle  j| \psi_2\rangle  \langle\psi_2 | k\rangle \nu^{d-r}(d\psi_2) =0.
\end{eqnarray*}
The last integral is zero for the same reason as in sub-case 1.

In conclusion all diagonal-real matrix elements are zero
$$
\bar{I}^{dr}_{ii;j,k} =0, \qquad {\mathrm for~all~}  i= 1,\dots ,2^k,  ~{\mathrm and} ~1\leq j<k \leq 2^k.
$$

\vspace{2mm}
\noindent
\emph{C) Diagonal-Imaginary block.}

Similarly to the case of diagonal-real elements, we obtain that all diagonal-imaginary matrix elements are zero
$$
\bar{I}^{di}_{ii;jk} =0, \qquad {\mathrm for~all~}  i= 1,\dots ,2^k,  ~{\mathrm and} ~1\leq j<k \leq 2^k.
$$

\vspace{2mm}
\noindent
\emph{D) Real-Real block.}
$$
I(\rho|{\bf B}_U)^{rr}_{ij;kl} = 4 \sum_{{\bf o}: p_\rho({\bf o}|{\bf B}_U)>0} 
\frac{1}{p_\rho({\bf o}|{\bf B}_U)}  
{\rm Re} ( \langle  i| {\bf o},U \rangle \langle{\bf o},U| j\rangle)\cdot {\rm Re} ( \langle  k| {\bf o},U \rangle \langle{\bf o},U| l\rangle)
$$
As before the corresponding entry of the average Fisher information matrix is
\begin{eqnarray*}
\bar{I}^{rr}_{ij;kl} &=&4\cdot  d\cdot r 
\int \frac{{\rm Re}( \langle  i| \psi \rangle \langle\psi | j\rangle ) \cdot {\rm Re}( \langle  k| \psi \rangle \langle\psi | l\rangle )}
{\sum_{k=1}^r | \langle \psi| k\rangle|^2} \nu^d(d\psi)\\
&=& 
4\cdot d\cdot r \int\!\!\!\!\int\!\!\!\!\int \frac{1}{q^2}   
{\rm Re}( \langle  i| \psi \rangle \langle\psi | j\rangle )  \cdot {\rm Re}( \langle  k| \psi \rangle \langle\psi | l\rangle ) m^{d,r}(dq)
 \nu^r(d\psi_1) \nu^{d-r}(d\psi_2).
\end{eqnarray*}
The same reasoning as before can be applied to transform the integral into a single integral, or a product of integrals over unitaries.
If a single index is larger than $r$ while the other 3 are smaller, or conversely a single index is smaller than $r$ while the other 
are larger, then we have two integrals over unitaries which are zero since the monomial is of odd order. Therefore the following cases remain to be analysed.

\noindent
\emph{Sub-case 1: $i ,j,k,l\leq r $.}
In this case the integral is
\begin{eqnarray*}
\bar{I}^{rr}_{ij;kl} &=& 4\cdot d\cdot r \int q^2  m^{d,r}(dq)
\int {\rm Re}( \langle  i| \psi_1 \rangle \langle\psi_1 | j\rangle ) \cdot {\rm Re}( \langle  k| \psi_1 \rangle \langle\psi_1 | l\rangle )   
\nu^r(d\psi_1) \\
&=&d\cdot r \cdot \frac{r}{d} \cdot 
\int \mu^r(dU) (U_{i1} U^*_{1j} +U_{j1} U^*_{1i}  ) ( U_{k1}U^*_{1l} +    U_{l1}U^*_{1k} ) .
\end{eqnarray*}
Using formula \eqref{eq.weingarten} we find that the integral over unitaries is zero unless we deal with a diagonal matrix element of $I$, i.e. $i= k$ and $j=l$. For the latter we have
\begin{eqnarray*}
\bar{I}^{rr}_{ij;ij} &=&
d\cdot r \cdot \frac{r}{d} \cdot 
\int \mu^r(dU) (U_{i1} U^*_{1j} +U_{j1} U^*_{1i}  ) ( U_{i1}U^*_{1j} +    U_{j1}U^*_{1i} ) \\
&=&
2 \cdot d\cdot r \cdot \frac{r}{d}  \left(
\frac{1}{r^2-1}- \frac{1}{r(r^2 -1)}  
\right)= \frac{2r}{r+1}.
\end{eqnarray*}
\noindent
\emph{Sub-case 2: $i ,j,k,l> r $.} In this case, the integral is 

\begin{eqnarray*}
\bar{I}^{rr}_{ij;kl} &=& 4\cdot d\cdot r \int \frac{(1-q^2)^2}{q^2}  m^{d,r}(dq)
\int {\rm Re}( \langle  i| \psi_2 \rangle \langle\psi_2 | j\rangle ) \cdot {\rm Re}( \langle  k| \psi_2 \rangle \langle\psi_2 | l\rangle )   
\nu^{d-r}(d\psi_2) \\
&=&d\cdot r \cdot  \frac{(d-r)(d-r+1)}{d(r-1)} \cdot 
\int \mu^{d-r}(dU) (U_{i1} U^*_{1j} +U_{j1} U^*_{1i}  ) ( U_{k1}U^*_{1l} +    U_{l1}U^*_{1k} ) .
\end{eqnarray*}
where we have used formula \eqref{eq.funny.integral} for the integral over $q$. A similar calculation as above 
shows that all off-diagonal elements are zero and 
\begin{eqnarray*}
\bar{I}^{rr}_{ij;ij} &=& 2 d\cdot r \cdot  \frac{(d-r)(d-r+1)}{d(r-1)} \cdot \left(\frac{1}{(d-r)^2-1} - \frac{1}{(d-r)( (d-r)^2 -1)}\right)\\
&=&\frac{2r}{r-1}.
\end{eqnarray*}
\noindent
\emph{Sub-case 3: ($i ,j\leq r$ and $k,l> r $) or ($i ,j> r $ and $ k,l\leq r $).}
In this case we deal with a product of integrals over unitaries of the form
$$
\int \mu^r(dU) U_{i1} U^*_{1j} =0.
$$
so all matrix elements are zero.

\noindent
\emph{Sub-case 4: $i ,k\leq r $ and $ j,l> r $.}
Again, the off-diagonal elements are zero and 
\begin{eqnarray*}
\bar{I}^{rr}_{ij;ij}= 
&=& 4\cdot d\cdot r \int (1-q^2)  m^{d,r}(dq)
\int {\rm Re}( \langle  i| \psi_1 \rangle \langle\psi_2 | j\rangle )\cdot  {\rm Re}( \langle  i| \psi_1 \rangle \langle\psi_2 | j\rangle )   
\nu^r(d\psi_1)\nu^r(d\psi_2)  \\
&=&2 d\cdot r \cdot (1-\frac{r}{d}) \cdot\frac{1}{r}\frac{1}{d-r}= 2.
\end{eqnarray*}

In conclusion, the only non-zero elements of the real-real block are on the diagonal and are given by
\begin{eqnarray*}
\bar{I}^{rr}_{ij;ij} &=&\frac{2 r}{r+1},\qquad  1 \leq i<j\leq r,\\
\bar{I}^{rr}_{ij;ij} &= &2,  \qquad\qquad~ i\leq r  ~{\mathrm and } ~r <j \leq 2^k,\\
\bar{I}^{rr}_{ij;ij} &= &\frac{2r}{r-1}, \qquad~ r< i<j\leq 2^k.
\end{eqnarray*}

\noindent
\emph{E) Imaginary-Imaginary block.}
This block is similar to the real-real one. All off diagonal elements are zero, and the diagonal ones are
\begin{eqnarray*}
\bar{I}^{ii}_{ij;ij}&=&  
\frac{2r}{r+1},\qquad  1 \leq i<j\leq r,\\
\bar{I}^{ii}_{ij;ij}&=&  
2, \qquad\qquad~ i\leq r  ~{\mathrm and } ~r <j \leq 2^k,\\
\bar{I}^{ii}_{ij;ij}&=&  
\frac{2r}{r-1},\qquad~ r< i<j\leq 2^k.\\
\end{eqnarray*}

\noindent
\emph{F) Real-Imaginary block.} Next we show that all real-imaginary off-diagonal elements are equal to zero.
$$
I(\rho|{\bf B}_U)^{ri}_{ij;kl} = - 4\cdot \sum_{{\bf o}: p_\rho({\bf o}|{\bf B}_U)>0} 
\frac{1}{p_\rho({\bf o}|{\bf B}_U)}  
{\rm Re} ( \langle  i| {\bf o},U \rangle \langle{\bf o},U| j\rangle){\rm Im} ( \langle  k| {\bf o},U \rangle \langle{\bf o},U| l\rangle)
$$
By integration we obtain the corresponding matrix element of the average Fisher information matrix
\begin{eqnarray*}
\bar{I}^{rr}_{ij;kl} &=& - 4\cdot d\cdot r 
\int \frac{{\rm Re}( \langle  i| \psi \rangle \langle\psi | j\rangle ) \cdot {\rm Im}( \langle  k| \psi \rangle \langle\psi | l\rangle )}
{\sum_{k=1}^r | \langle \psi| k\rangle|^2} \nu^d(d\psi)\\
&=& 
- 4\cdot d\cdot r \int \!\!\!\!\int \!\!\!\!\int \frac{1}{q^2}   
{\rm Re}( \langle  i| \psi \rangle \langle\psi | j\rangle )\cdot  {\rm Im}( \langle  k| \psi \rangle \langle\psi | l\rangle ) m^{d,r}(dq)
 \nu^r(d\psi_1) \nu^{d-r}(d\psi_2).
\end{eqnarray*}
Again, if one index is smaller that $r$ while the other three are larger, or otherwise, then the matrix element in zero. We analyse the remaining cases.
 
\noindent
\emph{Sub-case 1: $i ,j,k,l\leq r $.}
In this case the integral is
\begin{eqnarray*}
\bar{I}^{ri}_{ij;kl} &=& -4\cdot d\cdot r \int q^2  m^{d,r}(dq)
\int {\rm Re}( \langle  i| \psi_1 \rangle \langle\psi_1 | j\rangle )  {\rm Im}( \langle  k| \psi_1 \rangle \langle\psi_1 | l\rangle )   
\nu^r(d\psi_1) \\
&=& i\cdot d\cdot r \cdot \frac{r}{d} 
\int \mu^r(dU) (U_{i1} U^*_{1j} +U_{j1} U^*_{1i}  ) ( U_{k1}U^*_{1l} -    U_{l1}U^*_{1k} ) .
\end{eqnarray*}
Using formula \eqref{eq.weingarten} we find that the integral over unitaries is zero unless $i= k$ and $j=l$. However even in this case, two of the four terms are zero, and the other two cancel each other.

\noindent
\emph{Sub-case 2: $i ,j,k,l> r $.}
Here the integrals over the unitaries are similar as in \emph{case 1} above, with the difference that they taken with respect to 
$\mu^{d-r}(dU)$. Therefore the matrix elements are zero.

\noindent
\emph{Sub-case 3: ($i ,j\leq r$ and $k,l> r $) or ($i ,j> r , k,l\leq r $).}
In this case we deal with a product of integrals over unitaries of the form
$$
\int \mu^r(dU) U_{i1} U^*_{1j} =0.
$$
so all matrix elements are zero.

\noindent
\emph{Sub-case 4: $i ,k\leq r $ and $ j,l> r $ }
Again, the off-diagonal elements are zero and 
\begin{eqnarray*}
\bar{I}^{ri}_{ij;ij}= 
&=& - 4\cdot d\cdot r \int (1-q^2)  m^{d,r}(dq)
\int {\rm Re}( \langle  i| \psi_1 \rangle \langle\psi_2 | j\rangle )\cdot  {\rm Im}( \langle  i| \psi_1 \rangle \langle\psi_2 | j\rangle )   
\nu^r(d\psi_1)\nu^r(d\psi_2)  \\
&=& i\cdot d\cdot r \cdot \left(1-\frac{r}{d}\right) 
\int \int \left( U_{i1} V^*_{1j} + V_{j1}U^*_{1i}  \right) \left( U_{i1} V^*_{1j} - V_{j1} U^*_{1i}\right)\mu^r(dU)\mu^{d-r}(dV)\\
&=&0.
\end{eqnarray*}
In the last integral, two terms are zero and two have different signs and cancel each other. In conclusion all elements of the real-imaginary block are equal to zero.

\vspace{4mm}
\noindent \emph{Summary of the computation of $\bar{I}$}. We found that all off-diagonal blocks of $\bar{I}$ are zero
$$
\bar{I}^{dr}= \bar{I}^{di}= \bar{I}^{ri}= \bar{I}^{rd} = \bar{I}^{id} = \bar{I}^{dr} =0.
$$ 
Moreover the real and imaginary diagonal blocks are diagonal, equal to each other and have three distinct values depending on the position of the indices $i<j$ with respect to $r$: 
$$
\bar{I}^{rr/ii} = {\rm Diag}
\left\{
\begin{array}{cccc}
\bar{I}^{rr/ii}_{ij;ij}&=&  
 \frac{2r}{r+1}, & 1\leq i,j\leq r \\[1mm]
\bar{I}^{rr/ii}_{ij;ij}&=&  2, &  1\leq i\leq r, ~{\textrm and}~r<j\leq 2^k \\[1mm]
\bar{I}^{rr/ii}_{ij;ij}&=&  \frac{2r}{r-1} & r<i<j\leq 2^k.
\end{array}
\right.
$$ 
Finally, the  $d\times d$ block $\bar{I}^{dd}$ is not diagonal but has a simple form
$$
\bar{I}^{dd}_{ii;jj} =
\left\{
\begin{array}{cccc}
\bar{I}^{dd}_{ii;ii} &=& \frac{2r}{r+1}, & 1\leq i\leq r \\[1mm]
\bar{I}^{dd}_{ii;ii} &=& \frac{2r}{r-1}, &   r<i\leq 2^k \\[1mm]
\bar{I}^{dd}_{ii;jj} &=& \frac{r}{r+1}, & i,j\leq r , ~{\textrm and}~ i\neq j \\[1mm]
\bar{I}^{dd}_{ii;jj} &=& \frac{r}{r+1}, & 1\leq i\leq r ~{\textrm and} ~ r<j\leq 2^k\\[1mm]
\bar{I}^{dd}_{ii;jj} &=& \frac{r}{r+1}, &1\leq j\leq r ~{\textrm and} ~ r<i\leq 2^k \\[1mm]
\bar{I}^{dd}_{ii;jj} &=& \frac{r}{r-1} &    r<i,j , ~{\textrm and}~ i\neq j .\\[1mm]
\end{array}
\right.
$$

The model $\mathcal{R}_{d,r}$ can be seen (locally) as the restriction of the full unconstrained model $\mathcal{S}_{r,d}$ parametrised by $\theta$, to the subset of parameters $\tilde{\theta}$ which are real and imaginary parts of matrix elements $\rho_{i,j}$ with $i\leq r<j$. Therefore the corresponding average Fisher information matrix is equal to the corresponding block of $\bar{I}$, i.e. $\bar{\tilde{I}}= 2\mathbf{1}_{2r(d-r)}$.

\qed


\end{document}